\documentclass[final,3p]{elsarticle}
\usepackage{lineno,hyperref}
\usepackage{amsmath}
\usepackage{amssymb}
\usepackage{algorithmic}
\usepackage{algorithm}
\usepackage{caption}
\usepackage{mathtools}
\usepackage{changes}
\usepackage{multirow}
\usepackage{verbatim}
\usepackage{mathrsfs}
\usepackage{graphicx}
\usepackage{amsmath,amsthm,bm,mathrsfs}
\usepackage{subcaption}
\theoremstyle{plain}% Theorem-like structures provided by amsthm.sty
\newtheorem{theorem}{Theorem}[section]

\newtheorem{proposition}[theorem]{Proposition}

\theoremstyle{definition}

\theoremstyle{remark}

\modulolinenumbers[5]

\journal{ArXiv.org}

\begin{document}

\begin{frontmatter}

\title{From Data $H(j\omega_i)$ to Balanced Truncation Family: A Projection-based Non-intrusive Approach}

\author[uz]{Umair~Zulfiqar\corref{mycorrespondingauthor}}
\cortext[mycorrespondingauthor]{Corresponding author}
\ead{umair@yangtzeu.edu.cn}
\address[uz]{School of Electronic Information and Electrical Engineering, Yangtze University, Jingzhou, Hubei, 434023, China}
\begin{abstract}
Quadrature-based approximations of standard balanced truncation can be implemented non-intrusively within the Loewner framework using measured data, specifically transfer function samples on the imaginary axis of the $s$-plane. In contrast, quadrature-based approximations of several generalizations of balanced truncation—including positive-real, bounded-real, and stochastic balanced truncation—require samples of certain spectral factors on the imaginary axis, for which no practical measurement methods are available. This limitation was partially alleviated in \cite{zulfiqar2025data} by low-rank ADI-based approximations, which allowed some generalizations of balanced truncation to be implemented non-intrusively in the Loewner framework using transfer function samples. However, this approach requires samples in the open right-half of the $s$-plane. Such samples are impractical to obtain directly since exciting a physical system with unstable inputs can damage the system. It is argued in \cite{zulfiqar2025data} that the ADI shifts can be chosen in a close neighborhood of the imaginary axis, and that samples on the imaginary axis can then serve as surrogates for nearby samples in the open right-half plane. Alternatively, samples in a close neighborhood of the imaginary axis can be approximated from the Loewner quadruplet, since rational interpolation remains an admissible approximation in a close neighborhood of the interpolation points. However, this workaround lacks theoretical rigor.

To overcome these limitations, this paper presents data-driven implementations of balanced truncation and several of its generalizations that rely exclusively on transfer function samples on the imaginary axis. Rather than implicitly approximating the Gramians via numerical quadrature, the proposed approach approximates them implicitly through projection. This enables multiple members of the balanced truncation family to be implemented non-intrusively using practically measurable data, without requiring spectral factorizations. Using this projection-based framework, data-driven implementations are developed for standard balanced truncation, frequency-limited balanced truncation, time-limited balanced truncation, self-weighted balanced truncation, LQG balanced truncation, $H_{\infty}$ balanced truncation, positive-real balanced truncation, bounded-real balanced truncation, and stochastic balanced truncation. Numerical results demonstrate that the proposed non-intrusive implementations achieve performance comparable to their intrusive counterparts and accurately capture the dominant Hankel singular values.
\end{abstract}


\end{frontmatter}

%\linenumbers
\section{Introduction}
Consider a stable linear time-invariant system $H(s)$ of order $n$, given by a minimal state-space realization
\[
H(s) = C(sI - A)^{-1}B + D = G(s) + D,
\]
where $A \in \mathbb{R}^{n \times n}$, $B \in \mathbb{R}^{n \times m}$, $C \in \mathbb{R}^{p \times n}$, and $D \in \mathbb{R}^{p \times m}$.

The direct-feedthrough matrix $D$ satisfies
\[
D = \lim_{s \to \infty} H(s).
\]

Let $\tilde{H}(s)$ denote a stable reduced-order model (ROM) of order $r \ll n$ that approximates $H(s)$, with a minimal realization given by
\[
\tilde H(s) = \tilde C (sI - \tilde A)^{-1} \tilde B + D,
\]
where $\tilde A \in \mathbb{C}^{r \times r}$, $\tilde B \in \mathbb{C}^{r \times m}$, and $\tilde C \in \mathbb{C}^{p \times r}$.
\subsection{Loewner Framework \cite{mayo2007framework}}
Consider a set of right interpolation points \((\sigma_1,\dots,\sigma_{n_s})\) and a set of left interpolation points \((\mu_1,\dots,\mu_{n_u})\). In the rational interpolation \cite{beattie2017model}, the corresponding projection matrices are constructed as  
\begin{align}
V &= \begin{bmatrix} (\sigma_1 I - A)^{-1} B & \cdots & (\sigma_{n_s} I - A)^{-1} B\end{bmatrix},\label{Kry_V}\\
W &= \begin{bmatrix} (\mu_1^* I - A^T)^{-1} C^T & \cdots & (\mu_{n_u}^*I - A^T)^{-1} C^T\end{bmatrix}.\label{Kry_W}
\end{align}
Define the matrices $S_v$, $S_w$, $L_v$, and $L_w$ as:
\begin{align}
S_v &= \text{diag}(\sigma_1, \dots, \sigma_{n_s})\otimes I_m,& S_w &= \text{diag}(\mu_1, \dots, \mu_{n_u})\otimes I_p,\nonumber\\
L_v &= \begin{bmatrix} 1, \dots, 1 \end{bmatrix}\otimes I_m,& L_w^T &= \begin{bmatrix} 1, \dots, 1 \end{bmatrix}\otimes I_p.\label{SbLbScLc}
\end{align}
Then \(V\) and \(W\) satisfy the Sylvester equations:
\begin{align}
AV-VS_v+BL_v&=0,\label{sylv_V}\\
A^TW-WS_w^*+C^TL_w^T&=0.\label{sylv_W}
\end{align}
Next, define
\[
L=W^*V,\quad L_s=W^*AV,\quad \hat{B}=W^*B,\quad \hat{C}=CV.
\]
Then the ROM
\begin{align}
\tilde{H}(s)=\hat{C}(sL-L_s)^{-1}\hat{B}+D\label{Loew_ROM}
\end{align}
satisfies the interpolation conditions
\begin{align}
H(\sigma_i)= \tilde{H}(\sigma_i), \quad H(\mu_i) = \tilde{H}(\mu_i).\label{int_cond_1}
\end{align}
Moreover, whenever an interpolation point is shared by both sets, i.e., when \(\sigma_j = \mu_i\), the ROM also satisfies the Hermite interpolation condition
\begin{align}
H^\prime(\sigma_i) = \tilde{H}^\prime(\sigma_i).\label{int_cond_2}
\end{align}
The Loewner matrix $L$, the shifted Loewner matrix $L_s$, and the matrices $\hat{B}$ and $\hat{C}$ can be computed non-intrusively as:
\begin{align}
(L)_{ij} =&
\begin{cases}
-\frac{G(\sigma_j)-G(\mu_i)}{\sigma_j-\mu_i}, & \sigma_j \neq \mu_i, \\[4pt]
-G^{\prime}(\sigma_j),                & \sigma_j = \mu_i,
\end{cases}&
(L_s)_{ij}&=
\begin{cases}
-\frac{\sigma_jG(\sigma_j)-\mu_iG(\mu_i)}{\sigma_j-\mu_i}, & \sigma_j \neq \mu_i, \\[4pt]
-G(\sigma_j)-\sigma_jG^{\prime}(\sigma_j),                & \sigma_j = \mu_i,
\end{cases}\\
(\hat{B})_i&=G(\mu_i),& (\hat{C})_j&=G(\sigma_j).
\end{align}
Consequently, in the Loewner framework \cite{mayo2007framework}, the ROM in \eqref{Loew_ROM} can be constructed non-intrusively from samples of $H(s)$ at the points $\sigma_j$ and $\mu_i$, together with the sample $\lim_{s \to \infty} H(s)$.
\subsection{Quadrature-based Balanced Truncation \cite{goseaQuad}}
The controllability Gramian \(P\) and observability Gramian \(Q\) associated with the state-space realization \((A,B,C,D)\) admit frequency-domain integral representations:
\begin{align}
P &= \frac{1}{2\pi} \int_{-\infty}^{\infty} (j\omega I - A)^{-1} BB^T (-j\omega I - A^T)^{-1} \, d\omega, \label{int1}\\
Q &= \frac{1}{2\pi} \int_{-\infty}^{\infty} (-j\omega I - A^T)^{-1} C^T C (j\omega I - A)^{-1} \, d\omega. \label{int2}
\end{align}
These Gramians are the unique positive semidefinite solutions of the Lyapunov equations
\begin{align}
AP + PA^T + BB^T = 0,\label{lyap_P}\\
A^TQ + QA + C^T C = 0.\label{lyap_Q}
\end{align}
The integrals in \eqref{int1}–\eqref{int2} can be approximated using numerical integration, yielding
\begin{align}  
P &\approx\sum_{i=1}^{n_s} w_{p,i}^2 (j\omega_i I - A)^{-1} B B^T (-j\omega_i I - A^T)^{-1}, \label{quad_p} \\  
Q &\approx\sum_{i=1}^{n_u} w_{q,i}^2 (-j\nu_i I - A^T)^{-1} C^T C (j\nu_i I - A)^{-1}, \label{quad_q}
\end{align}  
where \(\omega_i,\nu_i\) are quadrature nodes and \(w_{p,i}^2, w_{q,i}^2\) the corresponding weights.

Define the diagonal scaling matrices
\begin{align}   
\tilde{L}_p &= \text{diag}(w_{p,1}, \dots, w_{p,{n_s}}) \otimes I_m, \nonumber \\  
\tilde{L}_q &= \text{diag}(w_{q,1}, \dots, w_{q,{n_u}}) \otimes I_p, \nonumber  
\end{align}
and set \(\sigma_i = j\omega_i\), \(\mu_i = j\nu_i\). Then, with \(V\) and \(W\) defined as in (\ref{Kry_V}) and (\ref{Kry_W}), the quadrature-based approximations \eqref{quad_p} and \eqref{quad_q} become
\[
P\approx (V \tilde{L}_p)(V \tilde{L}_p)^* \quad \text{and}\quad Q\approx (W \tilde{L}_q)(W \tilde{L}_q)^*.
\]
In quadrature-based balanced truncation (QuadBT) \cite{goseaQuad}, the exact Gramians of standard balanced truncation (BT) \cite{moore1981principal} are replaced by these quadrature-based approximations. Applying the square-root balancing procedure \cite{tombs1987truncated} yields the singular value decomposition
\begin{align} 
\tilde{L}_q^* L \tilde{L}_p = \begin{bmatrix}U_1 & U_2 \end{bmatrix} \begin{bmatrix} \Sigma_r & 0 \\ 0 & \Sigma_{n-r} \end{bmatrix} \begin{bmatrix} V_1^* \\ V_2^* \end{bmatrix}, \label{bsa_svd} 
\end{align} from which the projection matrices are computed as
\begin{align}
\tilde{W} = \tilde{L}_q U_1 \Sigma_r^{-1/2} \quad \text{and} \quad \tilde{V} = \tilde{L}_p V_1 \Sigma_r^{-1/2}.  \label{bsa_proj}
\end{align}
Finally, the ROM in QuadBT is obtained non-intrusively by projecting the Loewner quadruplet:
\begin{align}  
\tilde{A} &= \tilde{W}^* L_s \tilde{V}, & \tilde{B} &= \tilde{W}^* \hat{B}, & \tilde{C} &= \hat{C} \tilde{V}. \label{bsa_rom}
\end{align}
\subsection{Unresolved Issues in Data-driven Balanced Truncation Family}
In the literature, several alternative definitions of Gramians (and Gramian-like matrices) exist beyond those given in \eqref{int1} and \eqref{int2}. Many of these Gramians solve Riccati equations instead of Lyapunov equations \eqref{lyap_P} and \eqref{lyap_Q}. These Gramians replace standard Gramians $P$ and $Q$ leading to a family of balancing based model order reduction (MOR) algorithms.

We denote the Gramians associated with the BT family by \(\mathcal{P}\) and \(\mathcal{Q}\). Different members of the BT family are distinguished by the specific definitions of these Gramians. Data-driven variants of BT-family methods can be obtained provided that, as in QuadBT, the corresponding Gramians admit factorizations of the form
\begin{align}
\mathcal{P}\approx (V\tilde{Z}_p)(V\tilde{Z}_p)^*\quad \text{and}\quad \mathcal{Q}\approx (W\tilde{Z}_q)(W\tilde{Z}_q)^*,\label{gram_fact}
\end{align}where \(V\) and \(W\) are defined as in \eqref{Kry_V} and \eqref{Kry_W} with interpolation points \(\sigma_i=j\omega_i\) and \(\mu_i=j\nu_i\), while the factors \(\tilde{Z}_p\) and \(\tilde{Z}_q\) can be computed non-intrusively from data consisting of samples of \(H(s)\) at \(j\omega_i\) and \(j\nu_i\), together with the static gain \(\lim_{s\to\infty} H(s)\).

Recently, QuadBT has been extended to positive-real BT (PRBT) \cite{green1988balanced}, bounded-real BT (BRBT) \cite{opdenacker2002contraction}, and stochastic BT (BST) \cite{green1988balanced} in \cite{reiter2023generalizations}. In order to express the corresponding Gramians in integral forms analogous to \eqref{int1} and \eqref{int2}, the Gramians are represented as the controllability and observability Gramians associated with the state-space realizations of the spectral factorizations of \(H(s)H^{*}(s)\), \(I+H(s)H^{*}(s)\), and \(H(s)+H^{*}(s)\). As a consequence, the resulting non-intrusive formulations require samples of the corresponding spectral factors on the imaginary axis. Since no practical procedures exist for obtaining such spectral-factor data by injecting signals into a dynamical system and measuring its output response, these non-intrusive formulations remain largely of theoretical interest.

In \cite{zulfiqar2025data}, the Gramians arising in Linear–Quadratic–Gaussian BT (LQGBT) \cite{jonckheere1983new}, Self-weighted BT (SWBT) \cite{zhou1995frequency}, \(\mathcal{H}_\infty\) BT \cite{mustafa1991controller}, PRBT \cite{green1988balanced}, BRBT \cite{opdenacker2002contraction}, and BST \cite{green1988balanced} are approximated using low-rank alternating-direction implicit (ADI) methods for Lyapunov and Riccati equations \cite{benner2013efficient,wolf2016adi,benner2018radi}. These approximations admit factorizations of the form
\[
\mathcal{P}\approx (V\hat{Z}_p)(V\hat{Z}_p)^*\quad \text{and} \quad \mathcal{Q}\approx (W\hat{Z}_q)(W\hat{Z}_q)^*,
\]where \(V\) and \(W\) are defined as in \eqref{Kry_V} and \eqref{Kry_W} with interpolation points \(\sigma_i=\epsilon+j\omega_i\) and \(\mu_i=\epsilon+j\nu_i\) for some \(\epsilon>0\). The factors \(\hat{Z}_p\) and \(\hat{Z}_q\) can be computed non-intrusively from samples of \(H(s)\) evaluated at \(\epsilon+j\omega_i\) and \(\epsilon+j\nu_i\), together with the static gain \(\lim_{s\to\infty} H(s)\). From a theoretical standpoint, ADI methods require all shifts to have a nonzero real part, and hence \(\epsilon=0\) is not admissible. It is therefore argued in \cite{zulfiqar2025data} that the shifts may be chosen arbitrarily close to the imaginary axis and that samples \(H(j\omega_i)\) and \(H(j\nu_i)\) may serve as surrogates for \(H(\epsilon+j\omega_i)\) and \(H(\epsilon+j\nu_i)\). Alternatively, the values \(H(\epsilon+j\omega_i)\) and \(H(\epsilon+j\nu_i)\) may be approximated using a Loewner quadruplet constructed from samples on the imaginary axis, since rational interpolation remains a valid approximation in a neighborhood of the interpolation points. While this workaround is effective numerically, it lacks theoretical justification. 

The objective of this paper is to approximate the Gramians associated with the BT family, consistent with the factorization \eqref{gram_fact}, thereby ensuring that the resulting non-intrusive implementations are fully data-driven. Specifically, the proposed framework relies solely on samples of \(H(s)\) on the imaginary axis and the static gain \(\lim_{s\to\infty} H(s)\). The paper demonstrates that, when Gramian approximations are constructed via projection rather than numerical integration, the resulting approximations admit the desired factorization \eqref{gram_fact}, enabling data-driven implementations for several members of the BT family.
\section{Main Work}
In this section, we introduce projection-based approximations of \(\mathcal{P}\) and \(\mathcal{Q}\) in accordance with the factorization \eqref{gram_fact}. These approximations inherently perform rational interpolation with specific pole placement. To this end, we first discuss rational interpolation with enforcement of desired properties. These results are then repeatedly applied to construct projection-based approximations of \(\mathcal{P}\) and \(\mathcal{Q}\) according to \eqref{gram_fact} for various members of the BT family.
\subsection{Rational Interpolation with Desired Properties}
Let \(V\) and \(W\) be defined as in \eqref{Kry_V} and \eqref{Kry_W}, and let the matrices \(S_v\), \(L_v\), \(S_w\), and \(L_w\) be defined as in \eqref{SbLbScLc}. Further, assume that the pair \((-S_v, L_v)\) is observable and the pair \((-S_w, L_w)\) is controllable. Then, the ROM
\[
\tilde{H}(s)=CV(sI-S_v+\zeta_b L_v)^{-1}\zeta_b+D
\]
satisfies the interpolation condition
\begin{align}
H(\sigma_i)=\tilde{H}(\sigma_i),\label{int01}
\end{align}where \(\zeta_b \in \mathbb{C}^{mn_s \times m}\) is a free parameter; see \cite{wolf2014h,astolfi2010model,ahmad2011krylov}.

Dually, the ROM
\[
\tilde{H}(s)=\zeta_c(sI-S_w+ L_w\zeta_c)^{-1}W^*B+D
\]
satisfies the interpolation condition
\begin{align}
H(\mu_i)=\tilde{H}(\mu_i),\label{int02}
\end{align}where \(\zeta_c \in \mathbb{C}^{p \times pn_u}\) is a free parameter; see \cite{wolf2014h,astolfi2010model,ahmad2011krylov}.

The free parameters \(\zeta_b\) and \(\zeta_c\) can be chosen to enforce desired properties on the interpolant \(\tilde{H}(s)\). This principle forms the core of the results presented in this paper.
\subsubsection{Pole Placement}
The following proposition presents a method to construct a ROM that satisfies the interpolation condition \eqref{int01} with prescribed poles. This property will later be used to guarantee that the projected Lyapunov equation for approximating \(P\) has a unique solution.
\begin{proposition}\label{prop1}
Let \((\sigma_1, \dots, \sigma_{n_s})\) be \(n_s\) distinct interpolation points, and let \((\lambda_1, \dots, \lambda_{n_s})\) be \(n_s\) distinct desired poles such that the two sets have no elements in common. Define
\[
S_p = \mathrm{diag}(-\overline{\lambda_1}, \dots, -\overline{\lambda_{n_s}}) \otimes I_m.
\]
Further, assume that the pairs \((-S_v, L_v)\) and \((-S_p, L_v)\) are observable. Let \(X_p\) be the unique solution of the Sylvester equation
\begin{align}
-S_p^* X_p - X_p S_v + L_v^T L_v = 0. \label{Xp}
\end{align}
If \(V\) is defined as in \eqref{Kry_V}, then the interpolant
\begin{align}
\tilde{A} = S_v - \tilde{B} L_v = -X_p^{-1} S_p^* X_p, \quad
\tilde{B} = X_p^{-1} L_v^T, \quad
\tilde{C} = C V \label{ROM1}
\end{align}
satisfies the interpolation condition \eqref{int01} and has poles at \((\lambda_1, \dots, \lambda_{n_s})\) with multiplicity \(m\).
\end{proposition}
\begin{proof}
The assumption that the sets \((\sigma_1, \dots, \sigma_{n_s})\) and \((\lambda_1, \dots, \lambda_{n_s})\) are disjoint ensures the uniqueness of the Sylvester equation \eqref{Xp}. Moreover, the observability of the pairs \((-S_v, L_v)\) and \((-S_p, L_v)\) guarantees that \(X_p\) is invertible. Premultiplying \eqref{Xp} by \(X_p^{-1}\) gives
\begin{align}
-X_p^{-1} S_p^* X_p - S_v + X_p^{-1} L_v^T L_v = 0, \nonumber\\
S_v - \tilde{B} L_v = - X_p^{-1} S_p^* X_p, \nonumber\\
\tilde{A} = - X_p^{-1} S_p^* X_p. \nonumber
\end{align}
Hence, the eigenvalues of \(\tilde{A}\) coincide with those of \(-S_p^*\), and the ROM \eqref{ROM1} has the desired poles \((\lambda_1, \dots, \lambda_{n_s})\) with multiplicity \(m\).
\end{proof}
\begin{proposition}\label{prop2}
Assume that Proposition \ref{prop1} holds. Further, for \(\epsilon > 0\), let \(\sigma_i = j\omega_i\) and \(\lambda_i = -\epsilon + j\omega_i\), satisfying
\[\Delta_{\min} := \min_{\substack{i,j=1 \\ i \neq j}}^{n_s} |\omega_i - \omega_j| > 0.\]
Then \(X_p\) admits the explicit Kronecker form \(X_p = Y \otimes I_m\), where \(Y \in \mathbb{C}^{n_s \times n_s}\) has entries
\begin{equation}
Y_{ij} = \frac{1}{\epsilon + j(\omega_j - \omega_i)}, \qquad i,j = 1,\dots,n_s.
\end{equation}
The following properties hold:
\begin{enumerate}
  \item $X_p$ is asymptotically diagonal as $\epsilon \to 0^{+}$ with $X_p = \frac{1}{\epsilon} I + \mathcal{O}(1)$.
  \item By defining the diagonal approximation $\widetilde{X}_p := \frac{1}{\epsilon} I$, the relative error in Frobenius norm satisfies
\begin{equation}
\label{eq:relative_error}
\frac{\| X_p - \widetilde{X}_p \|_F}{\| \widetilde{X}_p \|_F} \leq \frac{\epsilon \sqrt{n_s - 1}}{\Delta_{\min}}.
\end{equation}
In particular, for any tolerance $\tau \in (0,1)$, if
\begin{equation}
\epsilon < \tau \, \frac{\Delta_{\min}}{\sqrt{n_s - 1}},
\end{equation}
then $\| X_p - \widetilde{X}_p \|_F / \| \widetilde{X}_p \|_F < \tau$.
  \item  $X_p$ is strictly diagonally dominant whenever
\begin{equation}
\epsilon < \frac{\Delta_{\min}}{n_s - 1}.
\end{equation}
\end{enumerate}
\end{proposition}
\begin{proof}
The proof is given in Appendix A.
\end{proof}
By selecting a small positive scalar \(\epsilon\) as in Proposition \ref{prop2}, we can effectively impose a modal structure with lightly damped poles on the ROM \eqref{ROM1}; that is, as \(\epsilon \to 0^+\),
\begin{align}
\tilde{A}\approx \mathrm{diag}(-\epsilon+j\omega_1,\cdots,-\epsilon+j\omega_{n_s})\otimes I_m, \quad \tilde{B}\approx \begin{bmatrix}\epsilon&\cdots&\epsilon\end{bmatrix}^T\otimes I_m.\label{modal}
\end{align}
The motivation for enforcing this modal form on the ROM's state-space matrices stems from well-established results in the MOR and control literature for flexible structures with lightly damped modes \cite{gawronski2004dynamics,gawronski2006balanced}. In such systems, the Gramians appearing in the BT family are block diagonally dominant, with off-diagonal blocks vanishing as \(\epsilon \to 0^+\). Consequently, the diagonal blocks of the associated Lyapunov, Sylvester, and Riccati equations can often be computed analytically, without solving full matrix equations. Computations of the matrix logarithm and matrix exponential also become straightforward. By imposing a modal form with lightly damped poles on the ROM, we aim to exploit these properties and avoid solving the projected matrix equations in the Petrov–Galerkin projection-based approximations developed in the sequel.
\subsubsection{Zero Placement}
The following proposition presents a method to construct a ROM that satisfies the interpolation condition \eqref{int01} with prescribed zeros. This property will later be used to approximate the Gramian involved in SWBT \cite{zhou1995frequency}.
\begin{proposition}\label{prop3}
Let \((\sigma_1, \dots, \sigma_{n_s})\) be \(n_s\) distinct interpolation points, and let \((\lambda_1, \dots, \lambda_{n_s})\) be \(n_s\) distinct desired zeros such that the two sets are disjoint. Define
\[
S_z=\mathrm{diag}(-\overline{\lambda_1},\cdots,-\overline{\lambda_{n_s}})\otimes I_m.
\]
Further, assume that the pairs \((-S_v, L_v + D^{-1} CV)\) and \((-S_z, L_v)\) are observable, and that the matrix \(D\) is invertible. Let \(V\) be defined as in \eqref{Kry_V}, and let \(X_z\) be the unique solution of the Sylvester equation
\begin{align}
-S_z^*X_z-X_zS_v+L_v^T(L_v+D^{-1}CV)=0.\label{Xz}
\end{align}
Then the interpolant that satisfies the interpolation condition \eqref{int01} and has zeros at \((\lambda_1, \dots, \lambda_{n_s})\) with multiplicity \(m\) is given by
\begin{align}
\tilde{A}=S_v-\tilde{B}L_v,\quad \tilde{B}=X_z^{-1}L_v^T,\quad \tilde{C}=CV.\label{ROM2}
\end{align}
\end{proposition}
\begin{proof}
The disjointness of the sets \((\sigma_1, \dots, \sigma_{n_s})\) and \((\lambda_1, \dots, \lambda_{n_s})\) ensures the uniqueness of the Sylvester equation \eqref{Xz}. Moreover, the observability of the pairs \((-S_v, L_v + D^{-1} CV)\) and \((-S_z, L_v)\) guarantees that \(X_z\) is invertible. Premultiplying \eqref{Xz} by \(X_z^{-1}\) gives
\begin{align}
-X_z^{-1} S_z^* X_z - S_v + X_z^{-1} L_v^T L_v + X_z^{-1} L_v^T D^{-1} CV &= 0, \nonumber\\
  S_v - \tilde{B} L_v - \tilde{B} D^{-1} \tilde{C} &= - X_z^{-1} S_z^* X_z, \nonumber\\
  \tilde{A} - \tilde{B} D^{-1} \tilde{C} &= - X_z^{-1} S_z^* X_z. \nonumber
  \end{align}
Hence, the eigenvalues of \(\tilde{A} - \tilde{B} D^{-1} \tilde{C}\) coincide with those of \(-S_z^*\), and the ROM \eqref{ROM2} has zeros at \((\lambda_1, \dots, \lambda_{n_s})\) with multiplicity \(m\).
\end{proof}
The next proposition establishes that \(X_z\) can become block-diagonally dominant under certain conditions.
\begin{proposition}\label{prop4}
Let Proposition \ref{prop3} hold. Further, for \(\epsilon > 0\), let \(\sigma_i = j\omega_i\) and \(\lambda_i = -\epsilon + j\omega_i\), satisfying
\[\Delta_{\mathrm{\min}} := \min_{\substack{i,j=1 \\ i \neq j}}^{n_s} |\omega_i - \omega_j| > 0.\]
Partition \(\tilde{C} = \begin{bmatrix} \tilde{C}_1 & \cdots & \tilde{C}_{n_s} \end{bmatrix}\) and define $M_j := I_m + D^{-1}\tilde{C}_{j}$ for $j=1,\dots,n_s$, assuming $\|M_j\| > 0$ for all $j$. Then $X_z = [X_{ij}]_{i,j=1}^{n_s}$ with $m\times m$ blocks $X_{ij}$ is row-wise block diagonally dominant:
\begin{equation*}
\|X_{ii}\| > \sum_{\substack{j=1\\ j\neq i}}^{n_s} \|X_{ij}\| \quad \text{for all } i=1,\dots,n_s,
\end{equation*}
provided that
\begin{equation*}
0 < \epsilon < \Delta_{\mathrm{min}} \cdot \min_{i=1,\dots,n_s} \frac{\|M_i\|}{\displaystyle\sum_{\substack{j=1\\ j\neq i}}^{n_s} \|M_j\|}.
\end{equation*}
Moreover, the relative off-diagonal weight in row $i$ satisfies
\begin{equation*}
\frac{\displaystyle\sum_{j\neq i} \|X_{ij}\|}{\|X_{ii}\|} \leq \frac{\epsilon}{\Delta_{\mathrm{min}}} \cdot \frac{\displaystyle\sum_{j\neq i} \|M_j\|}{\|M_i\|} = O(\epsilon) \quad \text{as } \epsilon \to 0^{+}.
\end{equation*}
\end{proposition}
\begin{proof}
The proof is given in Appendix B.
\end{proof}
\subsubsection{Shifted Pole Placement}
Rational interpolation with pole placement and zero placement discussed so far were special cases of a more general shifted pole placement problem, presented in the following proposition. Moreover, all projection-based approaches introduced later in the paper, as well as low-rank ADI methods for Lyapunov and Riccati equations \cite{benner2013efficient,wolf2016adi,benner2018radi}, can also be interpreted as special cases of this general framework. In particular, the low-rank ADI method for Lyapunov equations were shown to implicitly perform rational interpolation with pole placement in \cite{wolf2016adi}, while low-rank ADI method for Riccati equations were shown to perform rational interpolation with shifted pole placement in \cite{zulfiqar2025unified}.
\begin{proposition}\label{propg}
Let \((\sigma_1, \dots, \sigma_{n_s})\) be \(n_s\) distinct interpolation points, and let \(Q_{shift}\) be an \(mn_s \times mn_s\) matrix. Let \((\lambda_1, \dots, \lambda_{n_s})\) be \(n_s\) distinct desired poles, and define
\[
S_{sp}=\mathrm{diag}(-\overline{\lambda_1},\cdots,-\overline{\lambda_{n_s}})\otimes I_m.
\] Assume that \(S_{sp}\) and \(S_v - Q_{shift}\) have no eigenvalues in common, and that the pairs \((-S_v + Q_{shift}, L_v)\) and \((-S_{sp}, L_v)\) are observable. Let \(X_{sp}\) be the unique solution of the Sylvester equation
\begin{align}
-S_{sp}^*X_{sp}-X_{sp}(S_v-Q_{shift})+L_v^TL_v=0.\label{Xsp}
\end{align}
If \(V\) is defined as in \eqref{Kry_V}, the interpolant
\begin{align}
\tilde{A} = S_v - \tilde{B} L_v, \quad \tilde{B} = X_{sp}^{-1} L_v^T, \quad \tilde{C} = C V \label{ROMsp}
  \end{align}
satisfies the interpolation condition \eqref{int01}. Moreover, the matrix \(\tilde{A} - Q_{shift}\) has eigenvalues at \((\lambda_1, \dots, \lambda_{n_s})\) with multiplicity \(m\).
\end{proposition}
\begin{proof}
The assumption that \(S_{sp}\) and \(S_v - Q_{shift}\) have no common eigenvalues ensures the uniqueness of the Sylvester equation \eqref{Xsp}. Observability of the pairs \((-S_v + Q_{shift}, L_v)\) and \((-S_{sp}, L_v)\) guarantees that \(X_{sp}\) is invertible. Premultiplying \eqref{Xsp} by \(X_{sp}^{-1}\) gives
\begin{align}
-X_{sp}^{-1}S_{sp}^*X_{sp}-S_v+Q_{shift}+X_{sp}^{-1}L_v^TL_v=0,\nonumber\\
S_v-\tilde{B}L_v-Q_{shift}=-X_{sp}^{-1}S_{sp}^*X_{sp},\nonumber\\
\tilde{A}-Q_{shift}=-X_{sp}^{-1}S_{sp}^*X_{sp}.\nonumber
\end{align}
Hence, the eigenvalues of \(\tilde{A} - Q_{shift}\) coincide with those of \(-S_{sp}^*\), i.e., \((\lambda_1, \dots, \lambda_{n_s})\) with multiplicity \(m\).
\end{proof}
Let \( R_{shift} \) be an invertible matrix satisfying  
\[
Q_{shift} = \tilde{B} R_{shift} \tilde{C}
   = X_{sp}^{-1} L_v^T R_{shift} \tilde{C}.
\]
Then \( X_{sp} \) is the unique solution to the Sylvester equation  
\[
-S_{sp}^* X_{sp} - X_{sp} S_v + L_v^T\bigl(L_v + R_{shift} \tilde{C}\bigr) = 0.
\]
\begin{proposition}\label{propgm}
Let $Q_{shift}=\tilde{B}R_{shift}\tilde{C}$ and Proposition \ref{propg} hold. Further, for \(\epsilon > 0\), let \(\sigma_i = j\omega_i\) and \(\lambda_i = -\epsilon + j\omega_i\), satisfying
\[\Delta_{\mathrm{\min}} := \min_{\substack{i,j=1 \\ i \neq j}}^{n_s} |\omega_i - \omega_j| > 0.\]
Partition \(\tilde{C} = \begin{bmatrix} \tilde{C}_1 & \cdots & \tilde{C}_{n_s} \end{bmatrix}\) and define $M_j := I_m + R_{shift}\tilde{C}_{j}$ for $j=1,\dots,n_s$, assuming $\|M_j\| > 0$ for all $j$. Then $X_{sp} = [X_{ij}]_{i,j=1}^{n_s}$ with $m\times m$ blocks $X_{ij}$ is row-wise block diagonally dominant:
\begin{equation*}
\|X_{ii}\| > \sum_{\substack{j=1\\ j\neq i}}^{n_s} \|X_{ij}\| \quad \text{for all } i=1,\dots,n_s,
\end{equation*}
provided that
\begin{equation*}
0 < \epsilon < \Delta_{\mathrm{min}} \cdot \min_{i=1,\dots,n_s} \frac{\|M_i\|}{\displaystyle\sum_{\substack{j=1\\ j\neq i}}^{n_s} \|M_j\|}.
\end{equation*}
Moreover, the relative off-diagonal weight in row $i$ satisfies
\begin{equation*}
\frac{\displaystyle\sum_{j\neq i} \|X_{ij}\|}{\|X_{ii}\|} \leq \frac{\epsilon}{\Delta_{\mathrm{min}}} \cdot \frac{\displaystyle\sum_{j\neq i} \|M_j\|}{\|M_i\|} = O(\epsilon) \quad \text{as } \epsilon \to 0^{+}.
\end{equation*}
\end{proposition}
\begin{proof}
The proof is similar to that of Proposition \ref{prop4} and hence omitted for brevity.
\end{proof}
The next proposition provides conditions under which \( X_{sp} \) is block diagonally dominant for a general shift matrix \( Q_{\mathrm{shift}} \).  
\begin{proposition}\label{propgmg}
Let Proposition \ref{propg} hold. Let \(\epsilon > 0\), and let \(\sigma_i = j\omega_i\) and \(\lambda_i = -\epsilon + j\omega_i\), where
\[
\Delta_{\mathrm{min}} := \min_{\substack{i,j=1 \\ i \neq j}}^{n_s} |\omega_i - \omega_j| > 0.
\]
Suppose \(Q_{\mathrm{shift}} \in \mathbb{C}^{mn_s \times mn_s}\) is partitioned into \(m \times m\) blocks \(Q_{\mathrm{shift}} = [Q_{ij}]_{i,j=1}^{n_s}\), and \(X_{sp} = [X_{ij}]_{i,j=1}^{n_s}\) is likewise partitioned into \(m \times m\) blocks \(X_{ij} \in \mathbb{C}^{m \times m}\). Assume there exists an \(\eta\) in the interval \(\bigl(0, \frac{1}{6n_s}\bigr)\) such that
\[
\max_{1 \le i,j \le n_s} \|Q_{ij}\|_2 \le \eta \epsilon
\]
and that
\[
0 < \epsilon \le \frac{\Delta_{\min}}{2n_s}.
\]
Then, the matrix \(X_{sp}\) is strictly block diagonally dominant; i.e., for each \(i = 1, \dots, n_s\),
\[
\|X_{ii}\|_2 > \sum_{\substack{j=1 \\ j \neq i}}^{n_s} \|X_{ij}\|_2.
\]
\end{proposition}
\begin{proof}
The proof is given in Appendix C.
\end{proof}
Proposition \ref{propgmg} provides the theoretical basis for the projection-based approximations of \(\mathcal{P}\) and \(\mathcal{Q}\) developed in the sequel.
\subsection{Balanced Truncation}
Let us project the pair $(A,B)$ via Petrov–Galerkin projection
\begin{align}
\tilde{A}=W_p^*AV_p,\quad \mathrm{and}\quad \tilde{B}=W_p^*B,
\end{align}
satisfying the bi-orthogonality condition $W_{p}^* V_{p}=I$. Then solve the following projected Lyapunov equation
\begin{align}
\tilde{A}\tilde{P}+\tilde{P}\tilde{A}^*+\tilde{B}\tilde{B}^*=0.
\end{align}
The resulting approximation $V_p\tilde{P}V_p^*$ of the Gramian $P$ satisfies
\begin{align}
AV_{p}\tilde{P}V_{p}^*+V_{p}\tilde{P}V_{p}^*A^T+BB^T=R_{p},
\end{align}
with the residual fulfilling $W_{p}^* R_{p} W_{p}=0$.

Choosing $V_p = V$, where $V$ is defined in~\eqref{Kry_V} with interpolation points $\sigma_i = j\omega_i$, and selecting $(\tilde{A},\tilde{B})$ as in~\eqref{modal} according to Proposition~\ref{prop2}, eliminates the need to compute $W_p$ explicitly. As demonstrated in~\cite{gawronski2004dynamics,gawronski2006balanced}, when $\tilde{A}$ is in modal form with lightly damped modes $\lambda_i = -\epsilon + j\omega_i$, the controllability Gramian of $(\tilde{A},\tilde{B})$ becomes diagonally dominant, and its off-diagonal entries vanish as $\epsilon \to 0^{+}$. Consequently, for a sufficiently small $\epsilon > 0$ chosen as per Proposition~\ref{prop2}, the Gramian $P$ admits the approximation
\[ P\approx V\tilde{Z}_p\tilde{Z}_p^TV^*,\]
where $\tilde{Z}_p = \sqrt{\frac{\epsilon}{2}}\, I$, since the diagonal entries of $\tilde{P}$ equal $\frac{\epsilon}{2}$.

Dually, the observability Gramian $Q$ is approximated by
\[ Q\approx W\tilde{Z}_q\tilde{Z}_q^TW^*,\] with $W$ given in~\eqref{Kry_W} using interpolation points $\mu_i = j\nu_i$ on the imaginary axis and $\tilde{Z}_q = \sqrt{\frac{\epsilon}{2}}\, I$. Notably, similar to QuadBT, Proposition~\ref{prop2} enables the approximation of both $P$ and $Q$ without solving any projected Lyapunov equation. Although it remains unclear whether $\frac{\epsilon}{2}$ corresponds to a specific quadrature rule, Proposition~\ref{prop2} yields a quadrature-like approximation of the Gramians, facilitating a non-intrusive implementation of BT. The factors $\tilde{Z}_p$ and $\tilde{Z}_q$ are computable non-intrusively, consistent with the factorization in~\eqref{gram_fact}. Replacing $\tilde{L}_p$ and $\tilde{L}_q$ in QuadBT with $\tilde{Z}_p$ and $\tilde{Z}_q$, respectively, thus provides a non-intrusive, projection-based approximation of BT.
\subsection{Frequency-limited Balanced Truncation}
In~\cite{gawronski1990model}, a frequency-limited generalization of BT~\cite{moore1981principal} is introduced, enabling specification of a frequency interval over which enhanced approximation accuracy is desired. In frequency-limited BT (FLBT)~\cite{gawronski1990model}, the standard controllability and observability Gramians—defined over the entire frequency range—are replaced by their frequency-limited counterparts. Let $P_\Omega$ and $Q_\Omega$ denote the frequency-limited controllability and observability Gramians, respectively, defined over the symmetric frequency interval $[-\omega_2,-\omega_1]\cup[\omega_1,\omega_2]$ (in rad/sec) as
\begin{align}
P_\Omega&=\frac{1}{2\pi}\int_{\omega_1}^{\omega_2}(j\nu I-A)^{-1}BB^T(-j\nu I-A^T)^{-1}d\nu\nonumber\\
&\hspace*{3cm}+\frac{1}{2\pi}\int_{-\omega_2}^{-\omega_1}(j\nu I-A)^{-1}BB^T(-j\nu I-A^T)^{-1}d\nu,\label{Pw}\\
Q_\Omega&=\frac{1}{2\pi}\int_{\omega_1}^{\omega_2}(-j\nu I-A^T)^{-1}C^TC(j\nu I-A)^{-1}d\nu\nonumber\\
&\hspace*{3cm}+\frac{1}{2\pi}\int_{-\omega_2}^{-\omega_1}(-j\nu I-A^T)^{-1}C^TC(j\nu I-A)^{-1}d\nu.\label{Qw}
\end{align}
These Gramians satisfy the modified Lyapunov equations
\begin{align}
AP_\Omega+P_\Omega A^T+L_\Omega(A)BB^T+BB^TL_\Omega(A^T)&=0,\nonumber\\
A^TQ_\Omega+Q_\Omega A+L_\Omega(A^T)C^TC+C^TCL_\Omega(A)&=0,\nonumber
\end{align}
where
\begin{align}
L(A,\omega)&=\frac{1}{2\pi}\int_{-\omega}^{\omega}(j\nu I-A)^{-1}d\nu=\frac{j}{2\pi}\mathrm{ln}\Big((j\omega I+A)(-j\omega I+A)^{-1}\Big),\nonumber\\
L_\Omega(A)&=L(A,\omega_2)-L(A,\omega_1).\nonumber
\end{align}
To approximate \(P_\Omega\), we first select interpolation points \(\sigma_i = j\omega_i\) within the desired frequency interval \([-j\omega_2,-j\omega_1] \cup [j\omega_1,j\omega_2]\) and compute a ROM via Propositions \ref{prop1} and \ref{prop2} so that \(V(j\nu I-\tilde{A})^{-1}\tilde{B}\) interpolates \((j\nu I-A)^{-1}B\) at these points. Substituting \((j\nu I-A)^{-1}B\) with \(V(j\nu I-\tilde{A})^{-1}\tilde{B}\) yields the approximation
\begin{align}
P_\Omega&\approx V\Big(\frac{1}{2\pi}\int_{\omega_1}^{\omega_2}(j\nu I-\tilde{A})^{-1}\tilde{B}\tilde{B}^T(-j\nu I-\tilde{A}^*)^{-1}d\nu\Big)V^*\nonumber\\
&\hspace*{2cm}+V\Big(\frac{1}{2\pi}\int_{-\omega_2}^{-\omega_1}(j\nu I-\tilde{A})^{-1}\tilde{B}\tilde{B}^T(-j\nu I-\tilde{A}^*)^{-1}d\nu\Big)V^*.\label{Pwh}
\end{align}
By defining
\begin{align}
\tilde{P}_\Omega&=\frac{1}{2\pi}\int_{\omega_1}^{\omega_2}(j\nu I-\tilde{A})^{-1}\tilde{B}\tilde{B}^T(-j\nu I-\tilde{A}^*)^{-1}d\nu\nonumber\\
&\hspace*{3cm}+\frac{1}{2\pi}\int_{-\omega_2}^{-\omega_1}(j\nu I-\tilde{A})^{-1}\tilde{B}\tilde{B}^T(-j\nu I-\tilde{A}^*)^{-1}d\nu,
\end{align}
the approximation \eqref{Pwh} can be expressed as \(P_\Omega \approx V \tilde{P}_\Omega V^*\).

The seminal work \cite{gawronski1990model} not only introduced FLBT but also discussed efficient computation of the frequency-limited Gramians in modal form. Their diagonal dominance is established in~\cite{gawronski2006balanced}, as a consequence of the diagonal structure of the standard Gramians and the matrix logarithm. As shown in~\cite{gawronski1990model}, the relation
\[\tilde{P}_\Omega=L_\Omega(\tilde{A})\tilde{P}+\tilde{P}L_\Omega(\tilde{A}^*)\]
holds between the frequency-limited and infinite-frequency reduced Gramians. 

Since $\tilde{A}$ is nearly diagonal with eigenvalues $\lambda_i = -\epsilon + j\omega_i$, the matrix $L_\Omega(\tilde{A})$ is also nearly diagonal, because the matrix logarithm of a diagonal matrix remains diagonal~\cite{gawronski1990model}. Specifically,
\[
L_\Omega(\tilde{A})\approx\mathrm{diag}(L_\Omega(-\epsilon+j\omega_1),\cdots,L_\Omega(-\epsilon+j\omega_{n_s}))\otimes I_m.
\]Furthermore, because \(L_\Omega(\tilde{A}^*) = \overline{L_\Omega(\tilde{A})}\) \cite{gawronski1990model} and $\tilde{P}=\frac{\epsilon}{2}I$, we obtain
\[
\tilde{P}_\Omega\approx\tilde{Z}_p\tilde{Z}_p^T,
\]where
\[
\tilde{Z}_p=\mathrm{diag}\Big(\sqrt{\epsilon\mathrm{Re}\big(L_\Omega(-\epsilon+j\omega_1)\big)},\cdots,\sqrt{\epsilon\mathrm{Re}\big(L_\Omega(-\epsilon+j\omega_{n_s})\big)}\Big)\otimes I_m.
\]
It is not clear whether \(\epsilon\,\mathrm{Re}\big(L_\Omega(-\epsilon+j\omega_i)\big)\) corresponds to the weights of a quadrature rule. Nevertheless, \(\tilde{Z}_p\) can be computed non-intrusively.

Dually, the observability Gramian can be approximated as
\[
Q_\Omega\approx W\tilde{Z}_q\tilde{Z}_q^TW^*,
\]where \(W\) is defined as in \eqref{Kry_W} with interpolation points \(\mu_i = j\nu_i\) on the imaginary axis within the interval \([-j\omega_2,-j\omega_1] \cup [j\omega_1,j\omega_2]\), and
\[
\tilde{Z}_q=\mathrm{diag}\Big(\sqrt{\epsilon\mathrm{Re}\big(L_\Omega(-\epsilon+j\nu_1)\big)},\cdots,\sqrt{\epsilon\mathrm{Re}\big(L_\Omega(-\epsilon+j\nu_{n_u})\big)}\Big)\otimes I_p.
\]
Finally, a non-intrusive implementation of FLBT is achieved by replacing \(\tilde{L}_p\) and \(\tilde{L}_q\) in QuadBT with \(\tilde{Z}_p\) and \(\tilde{Z}_q\), respectively.
\subsection{Time-limited Balanced Truncation}
A time-limited generalization of BT was introduced in \cite{gawronski1990model}, allowing the user to specify a time interval where superior model accuracy is required. In time-limited BT (TLBT), the standard controllability and observability Gramians—defined over an infinite time horizon—are replaced by Gramians defined over a finite interval of interest.

Let \(P_\tau\) and \(Q_\tau\) denote the time-limited controllability and observability Gramians, respectively, defined over the interval \([t_1, t_2]\) (in seconds), where
\begin{align}
P_\tau&=\int_{t_1}^{t_2}e^{At}BB^Te^{A^Tt}dt,\label{Pt}\\
Q_\tau&=\int_{t_1}^{t_2}e^{A^Tt}C^TCe^{At}dt.\label{Qt}
\end{align}
These Gramians satisfy the Lyapunov equations
\begin{align}
AP_\tau+P_\tau A^T+e^{At_1}BB^Te^{A^Tt_1}-e^{At_2}BB^Te^{A^Tt_2}&=0,\nonumber\\
A^TQ_\tau+Q_\tau A+e^{A^Tt_1}C^TCe^{At_1}-e^{A^Tt_2}C^TCe^{At_2}&=0.\nonumber
\end{align}
To approximate \(P_\tau\), we first construct a projection-based approximation of \(e^{At} B\) using Propositions \ref{prop1} and \ref{prop2}:
\[
e^{At}B\approx Ve^{\tilde{A}t}\tilde{B}
\]where \(V\) is defined as in \eqref{Kry_V} with interpolation points \(\sigma_i = j\omega_i\) on the imaginary axis. Substituting this approximation into \eqref{Pt} yields
\begin{align}
P_\tau\approx V\Big(\int_{t_1}^{t_2}e^{\tilde{A}t}\tilde{B}\tilde{B}^Te^{\tilde{A}^*t}dt\Big)V^*.\label{Pth}
\end{align}
By defining
\[
\tilde{P}_\tau=\int_{t_1}^{t_2}e^{\tilde{A}t}\tilde{B}\tilde{B}^Te^{\tilde{A}^*t}dt,
\]
the approximation \eqref{Pth} can be written as \( P_\tau \approx V \tilde{P}_\tau V^* \).

The seminal work~\cite{gawronski1990model}, which first introduced TLBT, also discussed efficient computation of time-limited Gramians when $\tilde{A}$ is in modal form. Their diagonal dominance was established in \cite{gawronski2006balanced}, following directly from the diagonal dominance of the standard Gramians and properties of the matrix exponential. As shown in \cite{gawronski1990model}, the time-limited and infinite-horizon Gramians are related by
\[\tilde{P}_\tau=e^{\tilde{A}t_1}\tilde{P}e^{\tilde{A}^*t_1}-e^{\tilde{A}t_2}\tilde{P}e^{\tilde{A}^*t_2}.\]
Because \(\tilde{A}\) is nearly diagonal and the matrix exponential of a diagonal matrix is diagonal, \(e^{\tilde{A} t}\) is also nearly diagonal:
\[
e^{\tilde{A}t}\approx\mathrm{diag}\big(e^{(-\epsilon+j\omega_1)t},\cdots,e^{(-\epsilon+j\omega_{n_s})t}\big)\otimes I_m.
\]Moreover, since \(e^{\tilde{A}^* t} = \overline{e^{\tilde{A} t}}\) and $\tilde{P}=\frac{\epsilon}{2}I$, it follows that
\[
\tilde{P}_\tau\approx\tilde{Z}_p\tilde{Z}_p^T,
\]where
\begin{align}
\tilde{Z}_p&=\mathrm{diag}\Big(\sqrt{\frac{\epsilon}{2}\big(\|e^{(-\epsilon+j\omega_1)t_1}\|^2-\|e^{(-\epsilon+j\omega_1)t_2}\|^2\big)},\cdots,\nonumber\\
&\hspace*{4cm}\sqrt{\frac{\epsilon}{2}\big(\|e^{(-\epsilon+j\omega_{n_s})t_1}\|^2-\|e^{(-\epsilon+j\omega_{n_s})t_2}\|^2\big)}\Big)\otimes I_m.\nonumber
\end{align}
It is not clear whether the terms \(\frac{\epsilon}{2}\big( | e^{(-\epsilon + j\omega_i) t_1} |^2 - | e^{(-\epsilon + j\omega_i) t_2} |^2 \big)\) correspond to weights of a quadrature rule. Nevertheless, the matrix $\tilde{Z}_p$ can be computed non-intrusively.

Dually, \(Q_\tau\) can be approximated as
\[
Q_\tau\approx W\tilde{Z}_q\tilde{Z}_q^TW^*,
\] where \(W\) is as defined in \eqref{Kry_W} with interpolation points \(\mu_i = j\nu_i\) on the imaginary axis, and
\begin{align}
\tilde{Z}_q&=\mathrm{diag}\Big(\sqrt{\frac{\epsilon}{2}\big(\|e^{(-\epsilon+j\nu_1)t_1}\|^2-\|e^{(-\epsilon+j\nu_1)t_2}\|^2\big)},\cdots,\nonumber\\
&\hspace*{4cm}\sqrt{\frac{\epsilon}{2}\big(\|e^{(-\epsilon+j\nu_{n_u})t_1}\|^2-\|e^{(-\epsilon+j\nu_{n_u})t_2}\|^2\big)}\Big)\otimes I_p.\nonumber
\end{align}
Finally, a non-intrusive implementation of TLBT is obtained by replacing \(\tilde{L}_p\) and \(\tilde{L}_q\) in QuadBT with \(\tilde{Z}_p\) and \(\tilde{Z}_q\), respectively.
\subsection{Self-weighted Balanced Truncation}
In certain applications, the relative error \([H(s)]^{-1}(H(s)-\tilde{H}(s))\) is a theoretically more suitable approximation criterion than the absolute error \(H(s)-\tilde{H}(s)\). For instance, if a reduced-order controller designed for a reduced-order plant model must also perform well with the original high-order plant, the relative error becomes an appropriate metric; see \cite{obinata2012model} for details. This motivates SWBT \cite{zhou1995frequency}, a method that preserves the minimum-phase property of the original model.

Assume \(G(s)\) is a stable, minimum-phase, square system with an invertible \(D\) matrix. In SWBT, the controllability Gramian $P$ remains the same as in standard BT, but the observability Gramian $Q$ is replaced by a self-weighted observability Gramian \(Q_{sw}\), which solves the Lyapunov equation
\begin{align}
(A-BD^{-1}C)^TQ_{sw}+Q_{sw}(A-BD^{-1}C)+C^T(DD^T)^{-1}C=0.
\end{align}
Consider a Petrov–Galerkin projection of the triplet $(A,B,C)$ given by
\begin{align}
\tilde{A}=W_p^*AV_p,\quad \tilde{B}=W_p^*B,\quad \tilde{C}=CV_p,
\end{align}with the bi-orthogonality condition \(W_p^* V_p = I\). The projected Lyapunov equation then becomes
\begin{align}
(\tilde{A}-\tilde{B}D^{-1}\tilde{C})^*\tilde{Q}_{sw}+\tilde{Q}_{sw}(\tilde{A}-\tilde{B}D^{-1}\tilde{C})+\tilde{C}^*(DD^T)^{-1}\tilde{C}=0.\label{proj_sw}
\end{align}
The resulting approximation $W_p \tilde{Q}_{sw} W_p^*$ of $Q_{sw}$ satisfies the residual equation
\begin{align}
(A-BD^{-1}C)^TW_p\tilde{Q}_{sw}W_p^*+W_p\tilde{Q}_{sw}W_p^*(A-BD^{-1}C)+C^T(DD^T)^{-1}C=R_q\nonumber
\end{align} with the residual $R_q$ satisfying $V_p^* R_q V_p = 0$.

Choosing $W_p = W$, where $W$ is defined in~\eqref{Kry_W} using interpolation points $\mu_i = j\nu_i$ on the imaginary axis, a ROM can be constructed via the dual of Proposition~\ref{prop4}:
\begin{align}
\tilde{A}=S_w-L_w\tilde{C}, \quad \tilde{B}=W^*B,\quad \tilde{C}=L_w^TX_z^{-1},
\end{align}
with the following approximations holding for small $\epsilon > 0$:
\begin{align}
\tilde{A}-\tilde{B}D^{-1}\tilde{C}&\approx \mathrm{diag}(-\epsilon+j\nu_1,\cdots,-\epsilon+j\nu_{n_u})\otimes I_p,\nonumber\\ \tilde{C}&\approx\epsilon\begin{bmatrix}(I_p+G(j\nu_1)D^{-1})^{-1}&\cdots&(I_p+G(j\nu_{n_u})D^{-1})^{-1}\end{bmatrix},\nonumber\\
X_z&\approx\frac{1}{\epsilon}\mathrm{blkdiag}\big(I_p+G(j\nu_1)D^{-1},\cdots,I_p+G(j\nu_{n_u})D^{-1}\big).\nonumber
\end{align}
Owing to the block-diagonal dominance of the projected Lyapunov equation~\eqref{proj_sw}, the self-weighted observability Gramian admits the approximation
\[
Q_{sw}\approx W\tilde{Z}_q(W\tilde{Z}_q)^*,
\]
where
\[
\tilde{Q}_{sw}\approx\tilde{Z}_q\tilde{Z}_q^*=\frac{\epsilon}{2}\mathrm{blkdiag}\Big(\big((G(j\nu_1)+D)(G(j\nu_1)+D)^*\big)^{-1},\cdots,\big((G(j\nu_{n_u})+D)(G(j\nu_{n_u})+D)^*\big)^{-1}\Big).
\]
For single-input single-output (SISO) systems, $\tilde{Z}_q$ simplifies to
\[
\tilde{Z}_q=\mathrm{diag}\Bigg(\sqrt{\frac{\epsilon}{2}\big(D^2+\|G(j\nu_1)\|^2+2\mathrm{Re}(G(j\nu_1))\big)^{-1}},\cdots,\sqrt{\frac{\epsilon}{2}\big(D^2+\|G(j\nu_{n_u})\|^2+2\mathrm{Re}(G(j\nu_{n_u}))\big)^{-1}}\Bigg).
\]
Since the controllability Gramian in SWBT is the same as in standard BT, we have \(\tilde{Z}_p = \sqrt{\epsilon/2}I\). Finally, a non-intrusive implementation of SWBT is obtained by replacing \(\tilde{L}_p\) and \(\tilde{L}_q\) in QuadBT with \(\tilde{Z}_p\) and \(\tilde{Z}_q\), respectively.
\subsection{LQG Balanced Truncation}
In LQGBT, a ROM is constructed to preserve the dominant LQG characteristics—quantities analogous to the Hankel singular values in standard BT. To this end, the controllability Gramian \(P\) and observability Gramian \(Q\) from standard BT are replaced by the LQG Gramians \(P_{lqg}\) and \(Q_{lqg}\), respectively, which satisfy the Riccati equations
\begin{align}
AP_{lqg}+P_{lqg}A^T+BB^T-P_{lqg}C^TCP_{lqg}=0,\\
A^TQ_{lqg}+Q_{lqg}A+C^TC-Q_{lqg}BB^TQ_{lqg}=0.
\end{align}
Consider a Petrov–Galerkin projection of the triplet \((A, B, C)\) given by
\begin{align}
\tilde{A}=W_p^*AV_p,\quad \tilde{B}=W_p^*B,\quad \tilde{C}=CV_p,
\end{align}with the bi-orthogonality condition \(W_p^* V_p = I\).  

Let \(\tilde{P}_{lqg}\) be a stabilizing solution to the following projected Riccati equation:
\begin{align}
\tilde{A}\tilde{P}_{lqg}+\tilde{P}_{lqg}\tilde{A}^*+\tilde{B}\tilde{B}^*-\tilde{P}_{lqg}\tilde{C}^*\tilde{C}\tilde{P}_{lqg}=0.\label{proj_P_lqg}
\end{align}
The projection-based approximation \(V_p \tilde{P}_{lqg} V_p^*\) of \(P_{lqg}\) then satisfies the residual equation
\[
AV_p\tilde{P}_{lqg}V_p^*+V_p\tilde{P}_{lqg}V_p^*A^T+BB^T-V_p\tilde{P}_{lqg}\tilde{C}^*\tilde{C}\tilde{P}_{lqg}V_p^*=R_p
\] satisfying $W_p^*R_pW_p=0$.

If the triplet \((\tilde{A}, \tilde{B}, \tilde{C})\) is obtained via Proposition \ref{prop2} with \(V_p = V\), where \(V\) is defined as in \eqref{Kry_V} using interpolation points \(\sigma_i = j\omega_i\) on the imaginary axis, then \(\tilde{A}\) is nearly in modal form with lightly damped modes. Moreover, it is shown in \cite{gawronski2004dynamics,gawronski2006balanced} that for a state matrix in modal form with lightly damped modes, the stabilizing solution to the Riccati equation of form \eqref{proj_P_lqg} is block‑diagonally dominant, with diagonal blocks
\begin{align}
\tilde{P}_{lqg}&\approx \tilde{Z}_p\tilde{Z}_p^*=\epsilon\mathrm{blkdiag}\Big(\big(G^*(j\omega_1)G(j\omega_1)\big)^{-1}\big((I_m+G^*(j\omega_1)G(j\omega_1))^{\frac{1}{2}}-I_m\big),\cdots,\nonumber\\
&\hspace*{3cm}\big(G^*(j\omega_{n_s})G(j\omega_{n_s})\big)^{-1}\big((I_m+G^*(j\omega_{n_s})G(j\omega_{n_s}))^{\frac{1}{2}}-I_m\big)\Big).\nonumber
\end{align}
Consequently, the projection‑based approximation of \(P_{lqg}\) is
\[
P_{lqg}\approx (V\tilde{Z}_p)(V\tilde{Z}_p)^*.
\]
Dually, the projection‑based approximation of \(Q_{lqg}\) is given by
\[
Q_{lqg}\approx (W\tilde{Z}_q)(W\tilde{Z}_q)^*
\]where
\begin{align}
\tilde{Z}_q\tilde{Z}_q^*&=\epsilon\mathrm{blkdiag}\Big(\big(G(j\nu_1)G^*(j\nu_1)\big)^{-1}\big((I_p+G(j\nu_1)G^*(j\nu_1))^{\frac{1}{2}}-I_p\big),\cdots,\nonumber\\
&\hspace*{3cm}\big(G(j\nu_{n_u})G^*(j\nu_{n_u})\big)^{-1}\big((I_p+G(j\nu_{n_u})G^*(j\nu_{n_u}))^{\frac{1}{2}}-I_p\big)\Big),\nonumber
\end{align}and $W$ is defined in~\eqref{Kry_W} using interpolation points $\mu_i = j\nu_i$ on the imaginary axis.

For SISO systems, \(\tilde{Z}_p\) and \(\tilde{Z}_q\) simplify to
\begin{align}
\tilde{Z_p}=\mathrm{diag}\Bigg(\sqrt{\epsilon\frac{\sqrt{1+\|G(j\omega_1)\|^2}-1}{\|G(j\omega_1)\|^2}},\cdots,\sqrt{\epsilon\frac{\sqrt{1+\|G(j\omega_{n_s})\|^2}-1}{\|G(j\omega_{n_s})\|^2}}\Bigg),\nonumber\\
\tilde{Z_q}=\mathrm{diag}\Bigg(\sqrt{\epsilon\frac{\sqrt{1+\|G(j\nu_1)\|^2}-1}{\|G(j\nu_1)\|^2}},\cdots,\sqrt{\epsilon\frac{\sqrt{1+\|G(j\nu_{n_u})\|^2}-1}{\|G(j\nu_{n_u})\|^2}}\Bigg).\nonumber
\end{align}
Finally, a non‑intrusive implementation of LQGBT is obtained by replacing \(\tilde{L}_p\) and \(\tilde{L}_q\) in QuadBT with \(\tilde{Z}_p\) and \(\tilde{Z}_q\), respectively.
\subsection{$\mathcal{H}_\infty$ Balanced Truncation}
In \(\mathcal{H}_\infty\)-BT, a ROM is constructed to preserve the dominant \(\mathcal{H}_\infty\) characteristics—quantities analogous to the Hankel singular values in standard BT. To this end, the standard controllability and observability Gramians \(P\) and \(Q\) are replaced by the \(\mathcal{H}_\infty\) Gramians \(P_{\mathcal{H}_\infty}\) and \(Q_{\mathcal{H}_\infty}\), respectively, which satisfy the Riccati equations
\begin{align}
AP_{\mathcal{H}_\infty}+P_{\mathcal{H}_\infty}A^T+BB^T-(1-\gamma^2)P_{\mathcal{H}_\infty}C^TCP_{\mathcal{H}_\infty}=0,\\
A^TQ_{\mathcal{H}_\infty}+Q_{\mathcal{H}_\infty}A+C^TC-(1-\gamma^2)Q_{\mathcal{H}_\infty}BB^TQ_{\mathcal{H}_\infty}=0,
\end{align}where $\gamma>0$.

Similar to LQGBT, these Gramians can be approximated via Proposition \ref{prop2} as
\[P_{\mathcal{H}_\infty}\approx (V\tilde{Z}_p)(V\tilde{Z}_p)^*\quad \text{and}\quad Q_{\mathcal{H}_\infty}\approx (W\tilde{Z}_q)(W\tilde{Z}_q)^*,\]
where \(V\) and \(W\) are defined as in \eqref{Kry_V} and \eqref{Kry_W}, respectively, with interpolation points \(\sigma_i = j\omega_i\) and \(\mu_i = j\nu_i\), and
\begin{align}
\tilde{Z_p}\tilde{Z}_p^*&=\epsilon\mathrm{blkdiag}\Big(\big((1-\gamma^2)G^*(j\omega_1)G(j\omega_1)\big)^{-1}\big((I_m+(1-\gamma^2)G^*(j\omega_1)G(j\omega_1))^{\frac{1}{2}}-I_m\big),\cdots,\nonumber\\
&\hspace*{3cm}\big((1-\gamma^2)G^*(j\omega_{n_s})G(j\omega_{n_s})\big)^{-1}\big((I_m+(1-\gamma^2)G^*(j\omega_{n_s})G(j\omega_{n_s}))^{\frac{1}{2}}-I_m\big)\Big),\nonumber\\
\tilde{Z_q}\tilde{Z}_q^*&=\epsilon\mathrm{blkdiag}\Big(\big((1-\gamma^2)G(j\nu_1)G^*(j\nu_1)\big)^{-1}\big((I_p+(1-\gamma^2)G(j\nu_1)G^*(j\nu_1))^{\frac{1}{2}}-I_p\big),\cdots,\nonumber\\
&\hspace*{3cm}\big((1-\gamma^2)G(j\nu_{n_u})G^*(j\nu_{n_u})\big)^{-1}\big((I_p+(1-\gamma^2)G(j\nu_{n_u})G^*(j\nu_{n_u}))^{\frac{1}{2}}-I_p\big)\Big).\nonumber
\end{align}
For SISO systems, \(\tilde{Z}_p\) and \(\tilde{Z}_q\) simplify to
\begin{align}
\tilde{Z_p}=\mathrm{diag}\Bigg(\sqrt{\epsilon\frac{\sqrt{1+(1-\gamma^2)\|G(j\omega_1)\|^2}-1}{(1-\gamma^2)\|G(j\omega_1)\|^2}},\cdots,\sqrt{\epsilon\frac{\sqrt{1+(1-\gamma^2)\|G(j\omega_{n_s})\|^2}-1}{(1-\gamma^2)\|G(j\omega_{n_s})\|^2}}\Bigg),\nonumber\\
\tilde{Z_q}=\mathrm{diag}\Bigg(\sqrt{\epsilon\frac{\sqrt{1+(1-\gamma^2)\|G(j\nu_1)\|^2}-1}{(1-\gamma^2)\|G(j\nu_1)\|^2}},\cdots,\sqrt{\epsilon\frac{\sqrt{1+(1-\gamma^2)\|G(j\nu_{n_u})\|^2}-1}{(1-\gamma^2)\|G(j\nu_{n_u})\|^2}}\Bigg).\nonumber
\end{align}
Finally, a non‑intrusive implementation of $\mathcal{H}_\infty$ BT is obtained by replacing \(\tilde{L}_p\) and \(\tilde{L}_q\) in QuadBT with \(\tilde{Z}_p\) and \(\tilde{Z}_q\), respectively.
\subsection{Positive-real Balanced Truncation}
In PRBT, the controllability Gramian \(P\) and observability Gramian \(Q\) from standard BT are replaced by \(P_{pr}\) and \(Q_{pr}\), respectively. These are the stabilizing solutions to the following Riccati equations:
\begin{align}
\big(A-B(D+D^T)^{-1}C\big)P_{pr}&+P_{pr}\big(A-B(D+D^T)^{-1}C\big)^T\nonumber\\
&+B(D+D^T)^{-1}B^T+P_{pr}C^T(D+D^T)^{-1}CP_{pr}=0,\\
\big(A-B(D+D^T)^{-1}C\big)^TQ_{pr}&+Q_{pr}\big(A-B(D+D^T)^{-1}C\big)\nonumber\\
&+C^T(D+D^T)^{-1}C+Q_{pr}B(D+D^T)^{-1}B^TQ_{pr}=0.
\end{align}
Now consider a ROM constructed via Proposition \ref{propgm} with \(R_{shift}=(D+D^T)^{-1}\). This yields the following approximate relations:
\begin{align}
\tilde{A}-\tilde{B}(D+D^T)^{-1}\tilde{C}&\approx \mathrm{diag}\big(-\epsilon+j\omega_1,\cdots,-\epsilon+j\omega_{n_s}\big)\otimes I_m,\nonumber\\
\tilde{B}&\approx\epsilon\begin{bmatrix}\big(I_m+(D+D^T)^{-1}G(j\omega_1)\big)^{-1}\\\vdots\\\big(I_m+(D+D^T)^{-1}G(j\omega_{n_s})\big)^{-1}\end{bmatrix},\nonumber\\
X_{sp}&\approx\frac{1}{\epsilon}\mathrm{blkdiag}\big(I_m+(D+D^T)^{-1}G(j\omega_1),\cdots,I_m+(D+D^T)^{-1}G(j\omega_{n_s})\big).\nonumber
\end{align}
As in LQGBT and \(\mathcal{H}_\infty\) BT, the stabilizing solution to the projected Riccati equation
\begin{align}
\big(\tilde{A}-\tilde{B}(D+D^T)^{-1}\tilde{C}\big)\tilde{P}_{pr}&+\tilde{P}_{pr}\big(\tilde{A}-\tilde{B}(D+D^T)^{-1}\tilde{C}\big)^*\nonumber\\
&+\tilde{B}(D+D^T)^{-1}\tilde{B}^*+\tilde{P}_{pr}\tilde{C}^*(D+D^T)^{-1}\tilde{C}\tilde{P}_{pr}=0\nonumber
\end{align}
becomes block diagonally dominant when the triplet \((\tilde{A},\tilde{B},\tilde{C})\) is obtained via Proposition \ref{propgm}. Using the analytical formulas from \cite{gawronski2004dynamics,gawronski2006balanced} for the diagonal blocks of a block diagonally dominant stabilizing solution of the Riccati equation, we obtain
\[
\tilde{P}_{pr}\approx\tilde{Z}_p\tilde{Z}_p^*=\epsilon\mathrm{blkdiag}\Big((\alpha_{p,1})^{-1}\big(I_m-(I_m-\alpha_{p,1}\beta_{p,1})^{\frac{1}{2}}\big),\cdots,(\alpha_{p,n_s})^{-1}\big(I_m-(I_m-\alpha_{p,n_s}\beta_{p,n_s})^{\frac{1}{2}}\big)\Big),
\]where
\begin{align}
\alpha_{p,i}&=G^*(j\omega_i)(D+D^T)^{-1}G(j\omega_i),\nonumber\\
\beta_{p,i}&=\big(I_m+(D+D^T)^{-1}G(j\omega_i)\big)^{-1}(D+D^T)^{-1}\big(I_m+(D+D^T)^{-1}G(j\omega_i)\big)^{-*}.\nonumber
\end{align}
Thus, the projection-based approximation of \(P_{pr}\) is
\[
P_{pr}\approx (V\tilde{Z}_p)(V\tilde{Z}_p)^*,
\]where \(V\) is defined as in \eqref{Kry_V} using interpolation points \(\sigma_i = j\omega_i\) on the imaginary axis.

Dually, the projection-based approximation of \(Q_{pr}\) is
\[
Q_{pr}\approx (W\tilde{Z}_q)(W\tilde{Z}_q)^*
\]with
\begin{align}
\tilde{Z}_q\tilde{Z}_q^*&=\epsilon\mathrm{blkdiag}\Big((\alpha_{q,1})^{-1}\big(I_p-(I_p-\alpha_{q,1}\beta_{q,1})^{\frac{1}{2}}\big),\cdots,(\alpha_{q,n_u})^{-1}\big(I_p-(I_p-\alpha_{q,n_u}\beta_{q,n_u})^{\frac{1}{2}}\big)\Big),\nonumber\\
\alpha_{q,i}&=G(j\nu_i)(D+D^T)^{-1}G^*(j\nu_i),\nonumber\\
\beta_{q,i}&=\big(I_p+(D+D^T)^{-1}G^*(j\nu_i)\big)^{-1}(D+D^T)^{-1}\big(I_p+(D+D^T)^{-1}G^*(j\nu_i)\big)^{-*},\nonumber
\end{align}and \(W\) is defined in~\eqref{Kry_W} using interpolation points \(\mu_i = j\nu_i\) on the imaginary axis.

For SISO systems, \(\tilde{Z}_p\) and \(\tilde{Z}_q\) simplify to
\begin{align}
\tilde{Z}_p=\mathrm{diag}\Bigg(\sqrt{\epsilon\frac{1-\sqrt{1-\alpha_{p,1}\beta_{p,1}}}{\alpha_{p,1}}},\cdots,\sqrt{\epsilon\frac{1-\sqrt{1-\alpha_{p,n_s}\beta_{p,{n_s}}}}{\alpha_{p,{n_s}}}}\Bigg),\nonumber\\
\tilde{Z}_q=\mathrm{diag}\Bigg(\sqrt{\epsilon\frac{1-\sqrt{1-\alpha_{q,1}\beta_{q,1}}}{\alpha_{q,1}}},\cdots,\sqrt{\epsilon\frac{1-\sqrt{1-\alpha_{q,n_u}\beta_{q,{n_u}}}}{\alpha_{q,{n_u}}}}\Bigg).\nonumber
\end{align}
A non-intrusive implementation of PRBT is then obtained by replacing \(\tilde{L}_p\) and \(\tilde{L}_q\) in QuadBT with \(\tilde{Z}_p\) and \(\tilde{Z}_q\), respectively.
\subsection{Bounded-real Balanced Truncation}
In BRBT, the controllability and observability Gramians \(P\) and \(Q\) from standard BT are replaced by \(P_{br}\) and \(Q_{br}\), which are the stabilizing solutions of the following Riccati equations:
\begin{align}
\big(A+BD^T(I_p-DD^T)^{-1}C\big)P_{br}&+P_{br}\big(A+BD^T(I_p-DD^T)^{-1}C\big)^T\nonumber\\
&+B\big(I_m+D^T(I_p-DD^T)^{-1}D\big)B^T+P_{br}C^T(I_p-DD^T)^{-1}CP_{br}=0,\nonumber\\
\big(A+B(I_m-D^TD)^{-1}D^TC\big)^TQ_{br}&+Q_{br}\big(A+B(I_m-D^TD)^{-1}D^TC\big)\nonumber\\
&+C^T\big(I_p+D(I_m-D^TD)^{-1}D^T\big)C+Q_{br}B(I_m-D^TD)^{-1}B^TQ_{br}=0.\nonumber
\end{align}
A ROM is constructed via Proposition~\ref{propgm} by setting \(R_{\text{shift}} = D^T (I_p - D D^T)^{-1}\), yielding the approximations:
\begin{align}
\tilde{A}&+\tilde{B}D^T(I_p-DD^T)^{-1}\tilde{C}\approx \mathrm{diag}\big(-\epsilon+j\omega_1,\cdots,-\epsilon+j\omega_{n_s}\big)\otimes I_m,\nonumber\\
\tilde{B}&\approx\epsilon\begin{bmatrix}\big(I_m-D^T(I_p-DD^T)^{-1}G(j\omega_1)\big)^{-1}\\\vdots\\\big(I_m-D^T(I_p-DD^T)^{-1}G(j\omega_{n_s})\big)^{-1}\end{bmatrix},\nonumber\\
X_{sp}&\approx\frac{1}{\epsilon}\mathrm{blkdiag}\big(I_m-D^T(I_p-DD^T)^{-1}G(j\omega_1),\cdots,I_m-D^T(I_p-DD^T)^{-1}G(j\omega_{n_s})\big).\nonumber
\end{align}
As in PRBT, the stabilizing solution to the projected Riccati equation in BRBT,
\begin{align}
\big(\tilde{A}+\tilde{B}D^T(I_p-DD^T)^{-1}\tilde{C}\big)\tilde{P}_{br}&+\tilde{P}_{br}\big(\tilde{A}+\tilde{B}D^T(I_p-DD^T)^{-1}\tilde{C}\big)^*\nonumber\\
&+\tilde{B}(I_m+D^T(I_p-DD^T)^{-1}D)\tilde{B}^*+\tilde{P}_{br}\tilde{C}^*(I_p-DD^T)^{-1}\tilde{C}\tilde{P}_{br}=0\nonumber
\end{align}
becomes block diagonally dominant when \((\tilde{A}, \tilde{B}, \tilde{C})\) is obtained from Proposition~\ref{propgm}. Applying the analytical formulas from \cite{gawronski2004dynamics,gawronski2006balanced} for the diagonal blocks of a block diagonally dominant stabilizing solution of the Riccati equation gives
\[
\tilde{P}_{br}\approx\tilde{Z}_p\tilde{Z}_p^*=\epsilon\mathrm{blkdiag}\Big((\alpha_{p,1})^{-1}\big(I_m-(I_m-\alpha_{p,1}\beta_{p,1})^{\frac{1}{2}}\big),\cdots,(\alpha_{p,n_s})^{-1}\big(I_m-(I_m-\alpha_{p,n_s}\beta_{p,n_s})^{\frac{1}{2}}\big)\Big),
\]where
\begin{align}
\alpha_{p,i}&=G^*(j\omega_i)(I_p-DD^T)^{-1}G(j\omega_i),\nonumber\\
\beta_{p,i}&=\big(I_m-D^T(I_p-DD^T)^{-1}G(j\omega_i)\big)^{-1}\big(I_m+D^T(I_p+DD^T)^{-1}D\big)\big(I_m-D^T(I_p-DD^T)^{-1}G(j\omega_i)\big)^{-*}.\nonumber
\end{align}
Hence, the projection-based approximation of \(P_{br}\) is
\[
P_{br}\approx (V\tilde{Z}_p)(V\tilde{Z}_p)^*,
\]where \(V\) is defined as in \eqref{Kry_V} using interpolation points \(\sigma_i = j\omega_i\) on the imaginary axis.

Dually, the projection-based approximation of \(Q_{br}\) is
\[
Q_{br}\approx (W\tilde{Z}_q)(W\tilde{Z}_q)^*
\]where
\begin{align}
\tilde{Z}_q\tilde{Z}_q^*&=\epsilon\mathrm{blkdiag}\Big((\alpha_{q,1})^{-1}\big(I_p-(I_p-\alpha_{q,1}\beta_{q,1})^{\frac{1}{2}}\big),\cdots,(\alpha_{q,n_u})^{-1}\big(I_p-(I_p-\alpha_{q,n_u}\beta_{q,n_u})^{\frac{1}{2}}\big)\Big),\nonumber\\
\alpha_{q,i}&=G(j\nu_i)(I_m-D^TD)^{-1}G^*(j\nu_i),\nonumber\\
\beta_{q,i}&=\big(I_p-D(I_m-D^TD)^{-1}G^*(j\nu_i)\big)^{-1}\big(I_p+D(I_m+D^TD)^{-1}D^T\big)\big(I_p-D(I_m-D^TD)^{-1}G^*(j\nu_i)\big)^{-*},\nonumber
\end{align}and \(W\) is defined in~\eqref{Kry_W} using interpolation points \(\mu_i = j\nu_i\) on the imaginary axis.

For SISO systems, \(\tilde{Z}_p\) and \(\tilde{Z}_q\) simplify to
\begin{align}
\tilde{Z}_p=\mathrm{diag}\Bigg(\sqrt{\epsilon\frac{1-\sqrt{1-\alpha_{p,1}\beta_{p,1}}}{\alpha_{p,1}}},\cdots,\sqrt{\epsilon\frac{1-\sqrt{1-\alpha_{p,n_s}\beta_{p,{n_s}}}}{\alpha_{p,{n_s}}}}\Bigg),\nonumber\\
\tilde{Z}_q=\mathrm{diag}\Bigg(\sqrt{\epsilon\frac{1-\sqrt{1-\alpha_{q,1}\beta_{q,1}}}{\alpha_{q,1}}},\cdots,\sqrt{\epsilon\frac{1-\sqrt{1-\alpha_{q,n_u}\beta_{q,{n_u}}}}{\alpha_{q,{n_u}}}}\Bigg).\nonumber
\end{align}
Finally, a non-intrusive implementation of BRBT is obtained by replacing \(\tilde{L}_p\) and \(\tilde{L}_q\) in QuadBT with \(\tilde{Z}_p\) and \(\tilde{Z}_q\), respectively.
\subsection{Stochastic Balanced Truncation}
In BST, the controllability Gramian remains identical to that of standard BT, while the observability Gramian \(Q\) is replaced by \(Q_s\), defined as the stabilizing solution of the Riccati equation 
\begin{align}
\big(A-(PC^T+BD^T)(DD^T)^{-1}C\big)^TQ_s+Q_s\big(A-(&PC^T+BD^T)(DD^T)^{-1}C\big)+C^T(DD^T)^{-1}C\nonumber\\
&+Q_s(PC^T+BD^T)(DD^T)^{-1}(PC^T+BD^T)^TQ_s=0.\nonumber
\end{align}
To obtain a projection-based approximation of \(Q_s\), we first approximate \(P\). Let \(\hat{Q}_s\) denote an approximation of \(Q_s\) obtained by replacing \(P\) with its approximation \(P \approx V\big(\frac{\epsilon}{2}I\big)V^*\), where \(V\) is defined as in \eqref{Kry_V} using interpolation points \(\sigma_i = j\omega_i\) on the imaginary axis. Setting \(\hat{C} = CV\), \(\hat{Q}_s\) becomes the stabilizing solution to
\begin{align}
\Big(A-\big(\frac{\epsilon}{2}V\hat{C}^*+BD^T\big)(DD^T)^{-1}C\Big)^*\hat{Q}_s+\hat{Q}_s\Big(A&-\big(\frac{\epsilon}{2}V\hat{C}^*+BD^T\big)(DD^T)^{-1}C\Big)+C^T(DD^T)^{-1}C\nonumber\\
&+\hat{Q}_s\big(\frac{\epsilon}{2}V\hat{C}^*+BD^T\big)(DD^T)^{-1}\big(\frac{\epsilon}{2}V\hat{C}^*+BD^T\big)^*\hat{Q}_s=0.\nonumber
\end{align}
Now project \((A,B,C)\) via \(W\) by setting the free parameter \(\tilde{C}\) according to the dual of Proposition \ref{propgmg}, with
\[
Q_{shift}=\Big(\frac{\epsilon}{2}W^*V\hat{C}^*+W^*BD^T\Big)(DD^T)^{-1}\tilde{C}=\Big(\frac{\epsilon}{2}L\hat{C}^*+\tilde{B}D^T\Big)(DD^T)^{-1}\tilde{C},
\]
where \(W\) is defined as in \eqref{Kry_W} using interpolation points \(\mu_i = j\nu_i\) on the imaginary axis. This ROM satisfies the approximate relations
\begin{align}
\tilde{A}-Q_{shift}&=S_w-L_w\tilde{C}-Q_{shift}\approx \mathrm{diag}\big(-\epsilon+j\nu_1,\cdots,-\epsilon+j\nu_{n_u}\big)\otimes I_p,\nonumber\\
\tilde{C}&=L_w^TX_{sp}^{-1}\approx\epsilon\Bigg[\Bigg(I_p-\Big(\frac{G(j\omega_1)-G(j\nu_1)}{j\omega_1-j\nu_1}G^*(j\omega_1)-G(j\nu_1)D^T\Big)(DD^T)^{-1}\Bigg)^{-1}\cdots\nonumber\\
&\hspace*{3cm}\Bigg(I_p-\Big(\frac{G(j\omega_{n_u})-G(j\nu_{n_u})}{j\omega_{n_u}-j\nu_{n_u}}G^*(j\omega_{n_u})-G(j\nu_{n_u})D^T\Big)(DD^T)^{-1}\Bigg)^{-1}\Bigg],\nonumber\\
X_{sp}&\approx\frac{1}{\epsilon}\mathrm{blkdiag}\Bigg(I_p-\Big(\frac{G(j\omega_1)-G(j\nu_1)}{j\omega_1-j\nu_1}G^*(j\omega_1)-G(j\nu_1)D^T\Big)(DD^T)^{-1},\cdots\nonumber\\
&\hspace*{3cm}I_p-\Big(\frac{G(j\omega_{n_u})-G(j\nu_{n_u})}{j\omega_{n_u}-j\nu_{n_u}}G^*(j\omega_{n_u})-G(j\nu_{n_u})D^T\Big)(DD^T)^{-1}\Bigg).\label{ROM_bst}
\end{align}
If \(j\omega_i = j\nu_i\), the term \(\frac{G(j\omega_i)-G(j\nu_i)}{j\omega_i-j\nu_i}\) is replaced by \(G^\prime(j\nu_i)\).

As in PRBT and BRBT, the stabilizing solution to the projected Riccati equation in BST,
\begin{align}
\Big(\tilde{A}-\big(\frac{\epsilon}{2}L\hat{C}^*+\tilde{B}D^T\big)(DD^T)^{-1}\tilde{C}\Big)^*\tilde{Q}_s+\tilde{Q}_s\Big(\tilde{A}&-\big(\frac{\epsilon}{2}L\hat{C}^*+\tilde{B}D^T\big)(DD^T)^{-1}\tilde{C}\Big)+\tilde{C}^*(DD^T)^{-1}\tilde{C}\nonumber\\
&+\tilde{Q}_s\big(\frac{\epsilon}{2}L\hat{C}^*+\tilde{B}D^T\big)(DD^T)^{-1}\big(\frac{\epsilon}{2}L\hat{C}^*+\tilde{B}D^T\big)^*\tilde{Q}_s=0.\nonumber
\end{align}
becomes block diagonally dominant when the triplet \((\tilde{A},\tilde{B},\tilde{C})\) is obtained via the dual of Proposition \ref{propgmg}. Applying the analytical formulas from \cite{gawronski2004dynamics,gawronski2006balanced} for the diagonal blocks of a block diagonally dominant stabilizing solution of the Riccati equation yields
\[
\tilde{Q}_{s}\approx\tilde{Z}_q\tilde{Z}_q^*=\epsilon\mathrm{blkdiag}\Big((\alpha_{q,1})^{-1}\big(I_p-(I_p-\alpha_{q,1}\beta_{q,1})^{\frac{1}{2}}\big),\cdots,(\alpha_{q,n_u})^{-1}\big(I_p-(I_p-\alpha_{q,n_u}\beta_{q,n_u})^{\frac{1}{2}}\big)\Big),
\]where
\begin{align}
\alpha_{q,i}&=\Big(G(j\nu_i)D^T-\frac{\epsilon}{2}\frac{G(j\omega_i)-G(j\nu_i)}{j\omega_i-j\nu_i}G^*(j\omega_i)\Big)(DD^T)^{-1}\Big(G(j\nu_i)D^T-\frac{\epsilon}{2}\frac{G(j\omega_i)-G(j\nu_i)}{j\omega_i-j\nu_i}G^*(j\omega_i)\Big)^*,\nonumber\\
\beta_{q,i}&=\Bigg(I_p-\Big(\frac{\epsilon}{2}\frac{G(j\omega_i)-G(j\nu_i)}{j\omega_i-j\nu_i}G^*(j\omega_i)-G(j\nu_i)D^T\Big)(DD^T)^{-1}\Bigg)^{-*}(DD^T)^{-1}\times\nonumber\\
&\hspace*{5cm}\Bigg(I_p-\Big(\frac{\epsilon}{2}\frac{G(j\omega_i)-G(j\nu_i)}{j\omega_i-j\nu_i}G^*(j\omega_i)-G(j\nu_i)D^T\Big)(DD^T)^{-1}\Bigg)^{-1}.\label{alpha_beta}
\end{align}
If \(j\omega_i = j\nu_i\), \(\frac{G(j\omega_i)-G(j\nu_i)}{j\omega_i-j\nu_i}\) is replaced by \(G^\prime(j\nu_i)\).

For SISO systems, \(\tilde{Z}_q\) simplifies to
\begin{align}
\tilde{Z}_q=\mathrm{diag}\Bigg(\sqrt{\epsilon\frac{1-\sqrt{1-\alpha_{q,1}\beta_{q,1}}}{\alpha_{q,1}}},\cdots,\sqrt{\epsilon\frac{1-\sqrt{1-\alpha_{q,n_u}\beta_{q,{n_u}}}}{\alpha_{q,{n_u}}}}\Bigg).\nonumber
\end{align}
Since the controllability Gramian in BST is the same as in standard BT, we have \(\tilde{Z}_p = \sqrt{\epsilon/2}\,I\). A non-intrusive implementation of BST is then obtained by replacing \(\tilde{L}_p\) and \(\tilde{L}_q\) in QuadBT with \(\tilde{Z}_p\) and \(\tilde{Z}_q\), respectively.
\section{Numerical Results}
This section evaluates the numerical performance of the proposed data-driven algorithms for BT-family using the same $400^{th}$-order RLC circuit model from \cite{reiter2023generalizations}. MATLAB code to reproduce the results is available in \cite{mycode}. All simulations were run in MATLAB R2021b on a laptop with a 2 GHz Intel i7 processor and 16 GB RAM.

The sampling points $j\omega_i$ are $50$ logarithmically spaced frequencies from $10^{-1}$ to $10^3$, with negative counterparts $-j\omega_i$ also included. The free parameter $\epsilon$ is set to $10^{-4}$. The sampling points $j\mu_i$ are identical to $j\sigma_i$. Transfer function samples are obtained numerically from the state-space realization of the RLC model in \citep{reiter2023generalizations}. The quantities $\sqrt{\lambda_i(P_\Omega Q_\Omega)}$, $\sqrt{\lambda_i(P_\tau Q_\tau)}$, $\sqrt{\lambda_i(PQ_{sw})}$, $\sqrt{\lambda_i(P_{lqg}Q_{lqg})}$, $\sqrt{\lambda_i(P_{\mathcal{H}_\infty}Q_{\mathcal{H}_\infty})}$, $\sqrt{\lambda_i(P_{pr}Q_{pr})}$, $\sqrt{\lambda_i(P_{br}Q_{br})}$, and $\sqrt{\lambda_i(PQ_{s})}$ are referred to as Hankel-like singular values in this section.

The time interval of interest in TLBT is set to \([0, 5]\) sec. Figures \ref{fig1}–\ref{fig8} compare the Hankel-like singular values and the relative error $\frac{\|G(s)-\hat{G}(s)\|_{\mathcal{H}_\infty}}{\|G(s)\|_{\mathcal{H}_\infty}}$. The $25^{th}$-order ROMs from intrusive methods (BT, TLBT, SWBT, LQGBT, $\mathcal{H}_\infty$ BT, PRBT, BRBT, and BST) and their non-intrusive counterparts accurately capture the $20$ most dominant Hankel-like singular values. The proposed data-driven methods also achieve accuracy comparable to intrusive approaches for ROMs of orders $1$ through $25$.
\begin{figure}[!h]
    \centering
    \begin{subfigure}{0.48\textwidth}
        \centering
        \includegraphics[width=\linewidth]{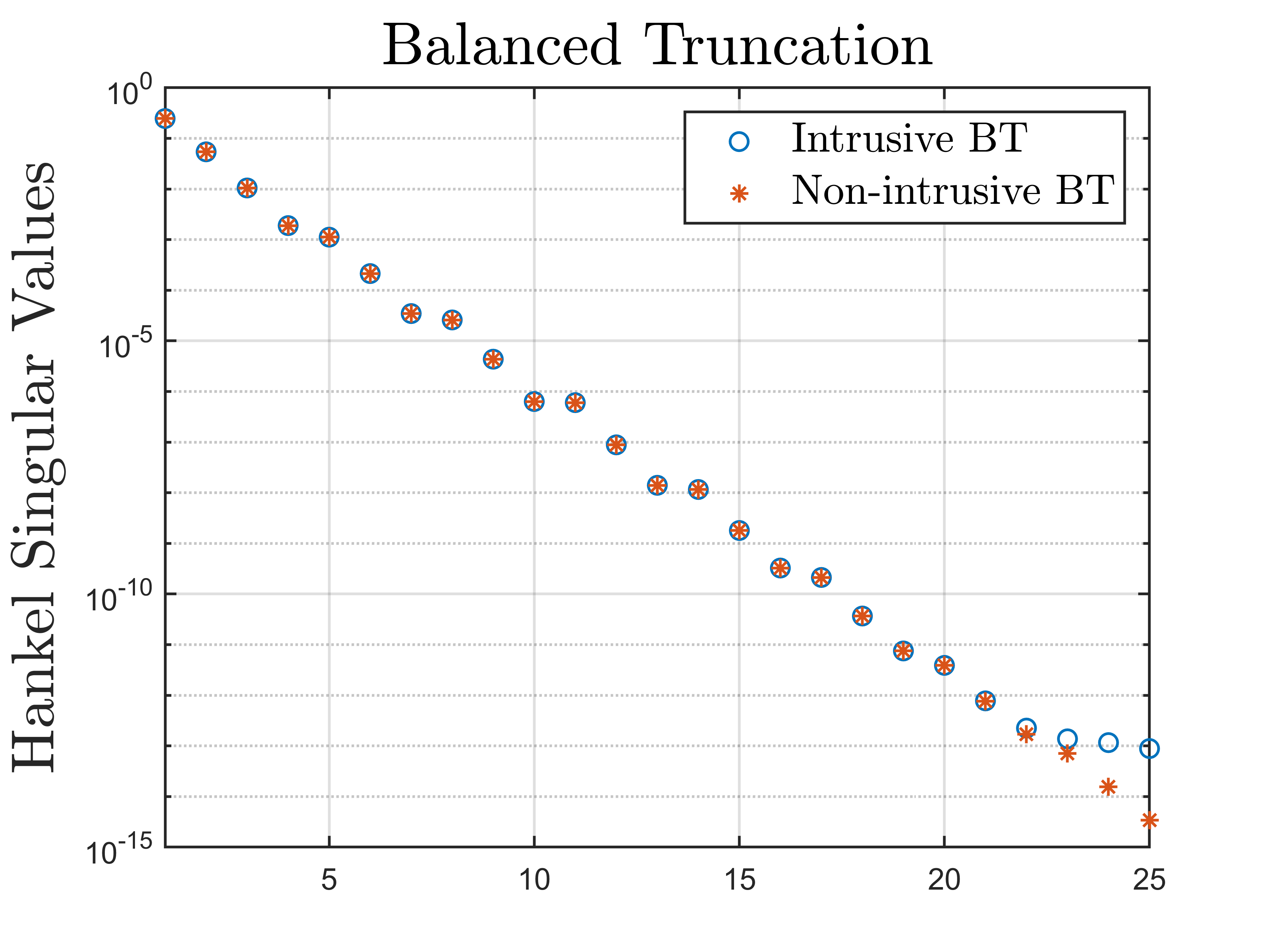}
        \caption{Hankel Singular Values Comparison}
    \end{subfigure}
    \hfill
    \begin{subfigure}{0.48\textwidth}
        \centering
        \includegraphics[width=\linewidth]{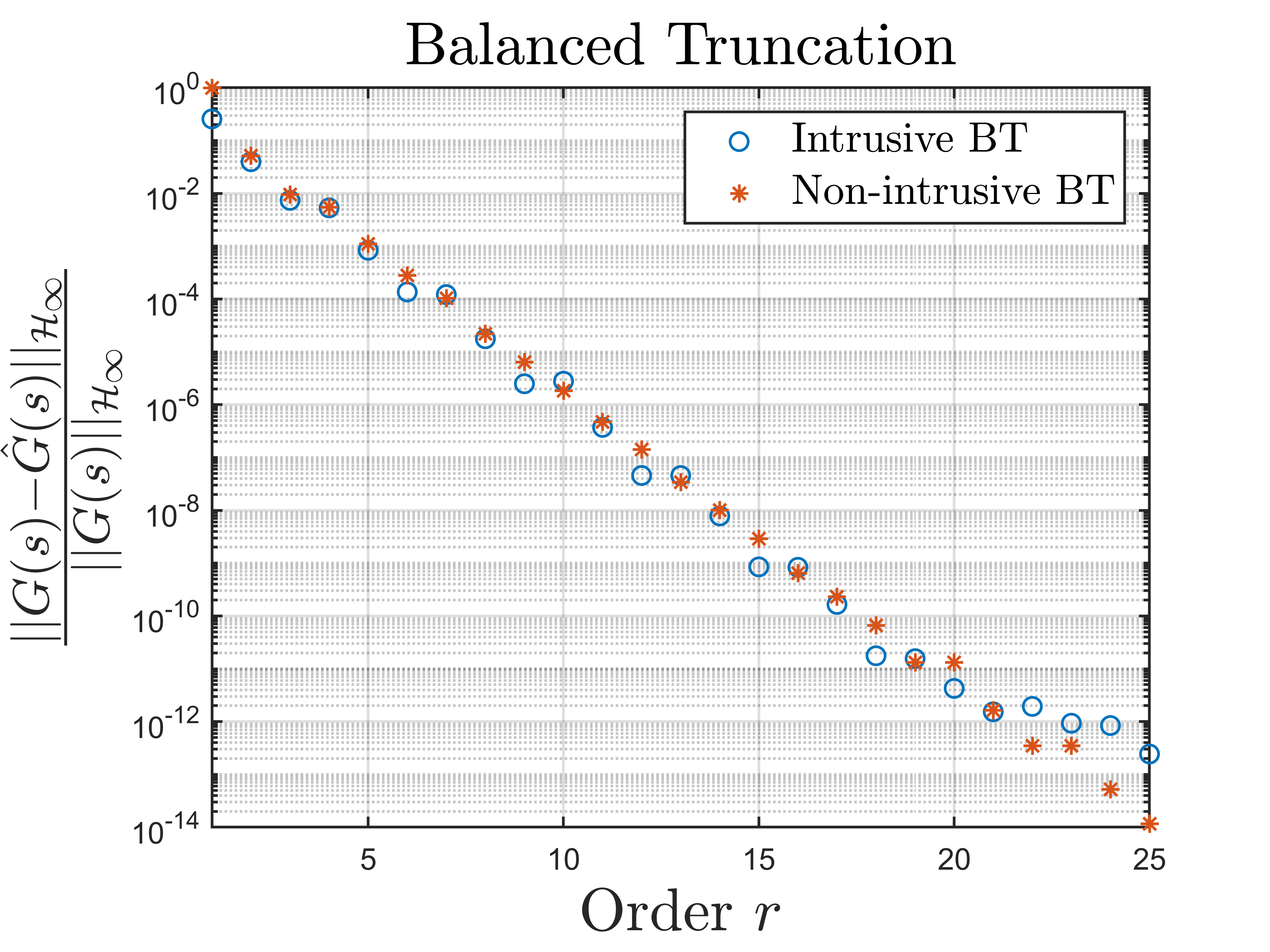}
        \caption{Relative Error Comparison}
    \end{subfigure}
    \caption{Performance Comparison between Intrusive and Non-intrusive BT}\label{fig1}
\end{figure}
\begin{figure}[!h]
    \centering
    \begin{subfigure}{0.48\textwidth}
        \centering
        \includegraphics[width=\linewidth]{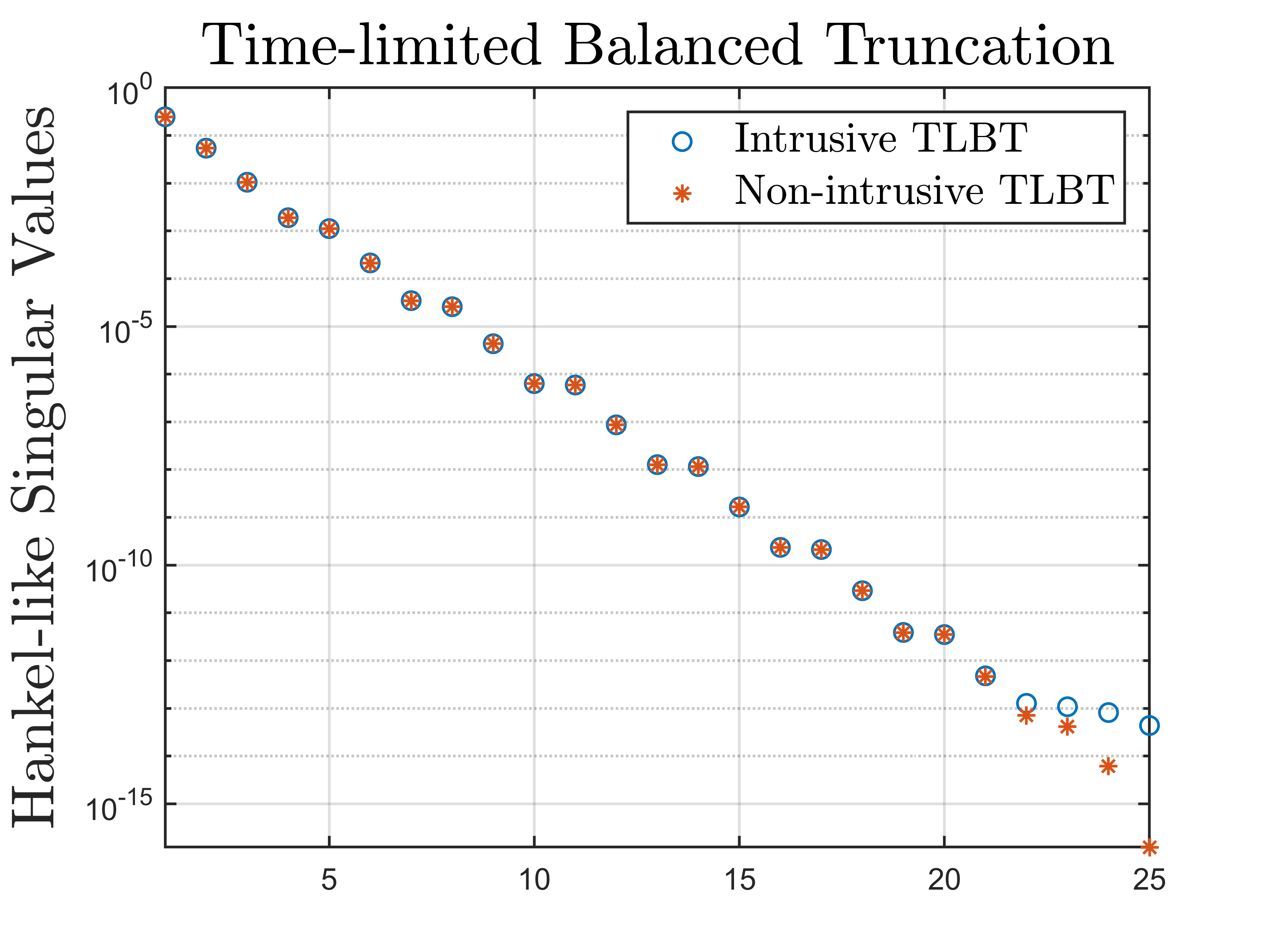}
        \caption{Hankel-like Singular Values Comparison}
    \end{subfigure}
    \hfill
    \begin{subfigure}{0.48\textwidth}
        \centering
        \includegraphics[width=\linewidth]{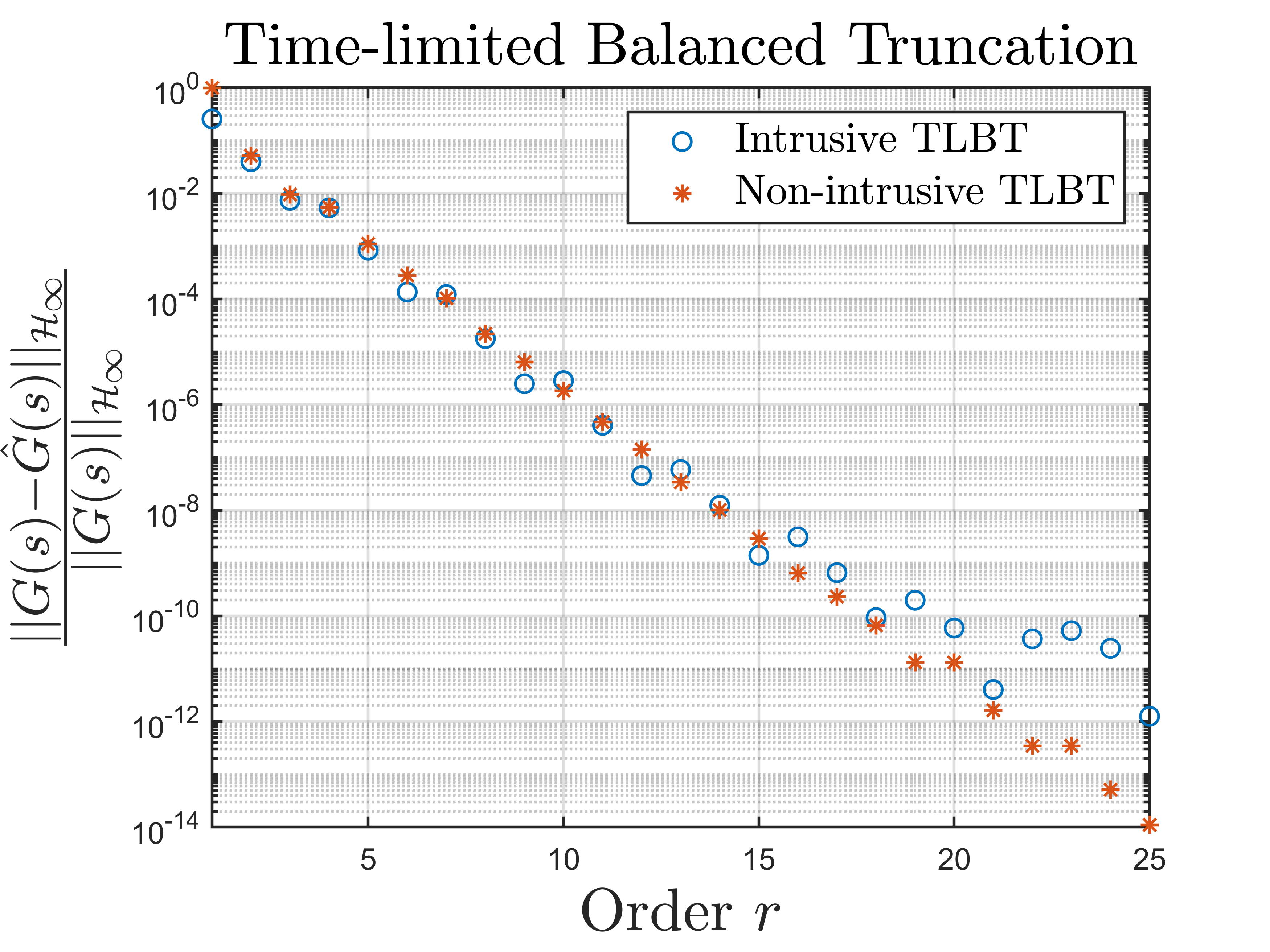}
        \caption{Relative Error Comparison}
    \end{subfigure}
    \caption{Performance Comparison between Intrusive and Non-intrusive TLBT}\label{fig2}
\end{figure}
\begin{figure}[!h]
    \centering
    \begin{subfigure}{0.48\textwidth}
        \centering
        \includegraphics[width=\linewidth]{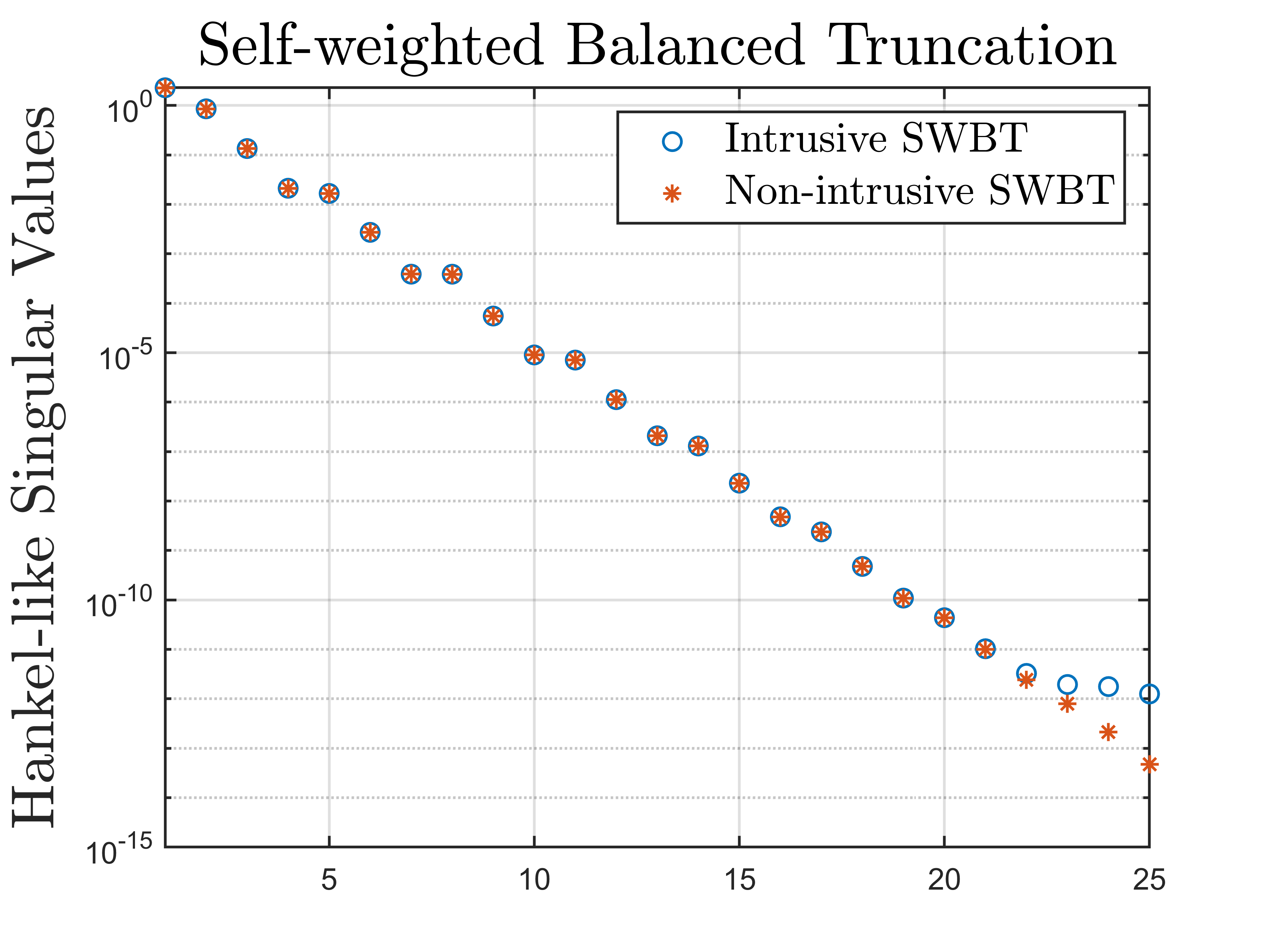}
        \caption{Hankel-like Singular Values Comparison}
    \end{subfigure}
    \hfill
    \begin{subfigure}{0.48\textwidth}
        \centering
        \includegraphics[width=\linewidth]{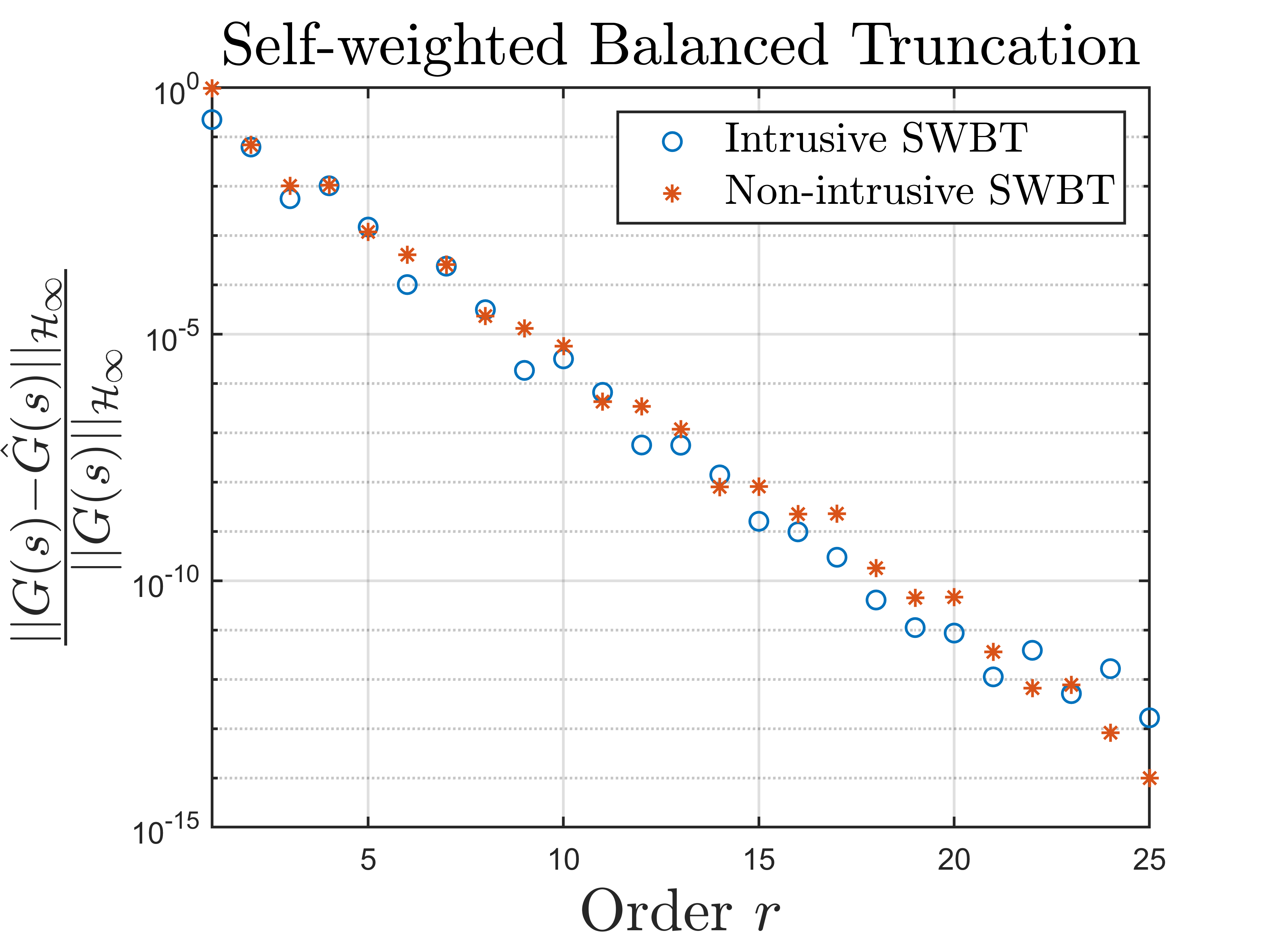}
        \caption{Relative Error Comparison}
    \end{subfigure}
    \caption{Performance Comparison between Intrusive and Non-intrusive SWBT}\label{fig3}
\end{figure}
\begin{figure}[!h]
    \centering
    \begin{subfigure}{0.48\textwidth}
        \centering
        \includegraphics[width=\linewidth]{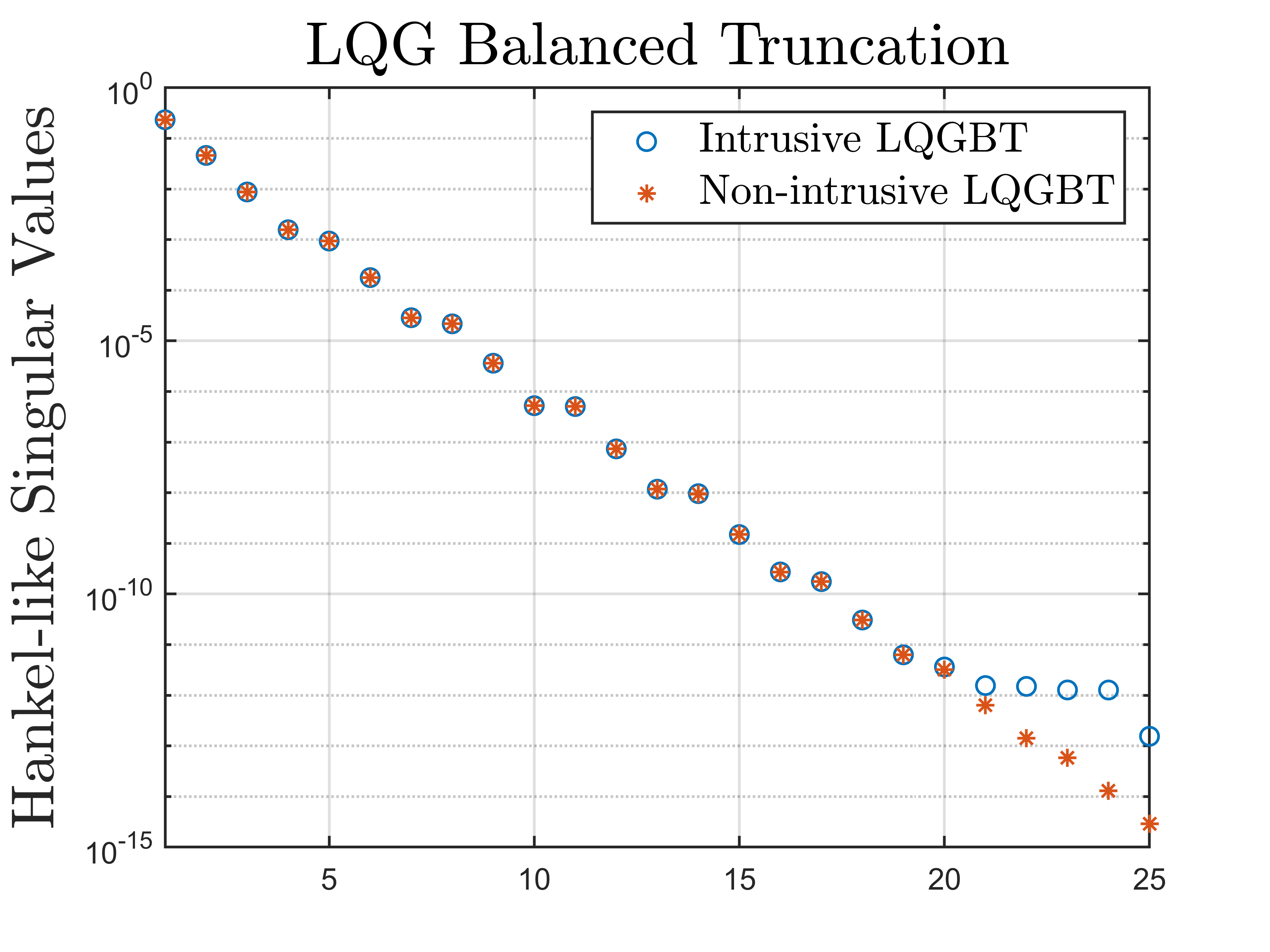}
        \caption{Hankel-like Singular Values Comparison}
    \end{subfigure}
    \hfill
    \begin{subfigure}{0.48\textwidth}
        \centering
        \includegraphics[width=\linewidth]{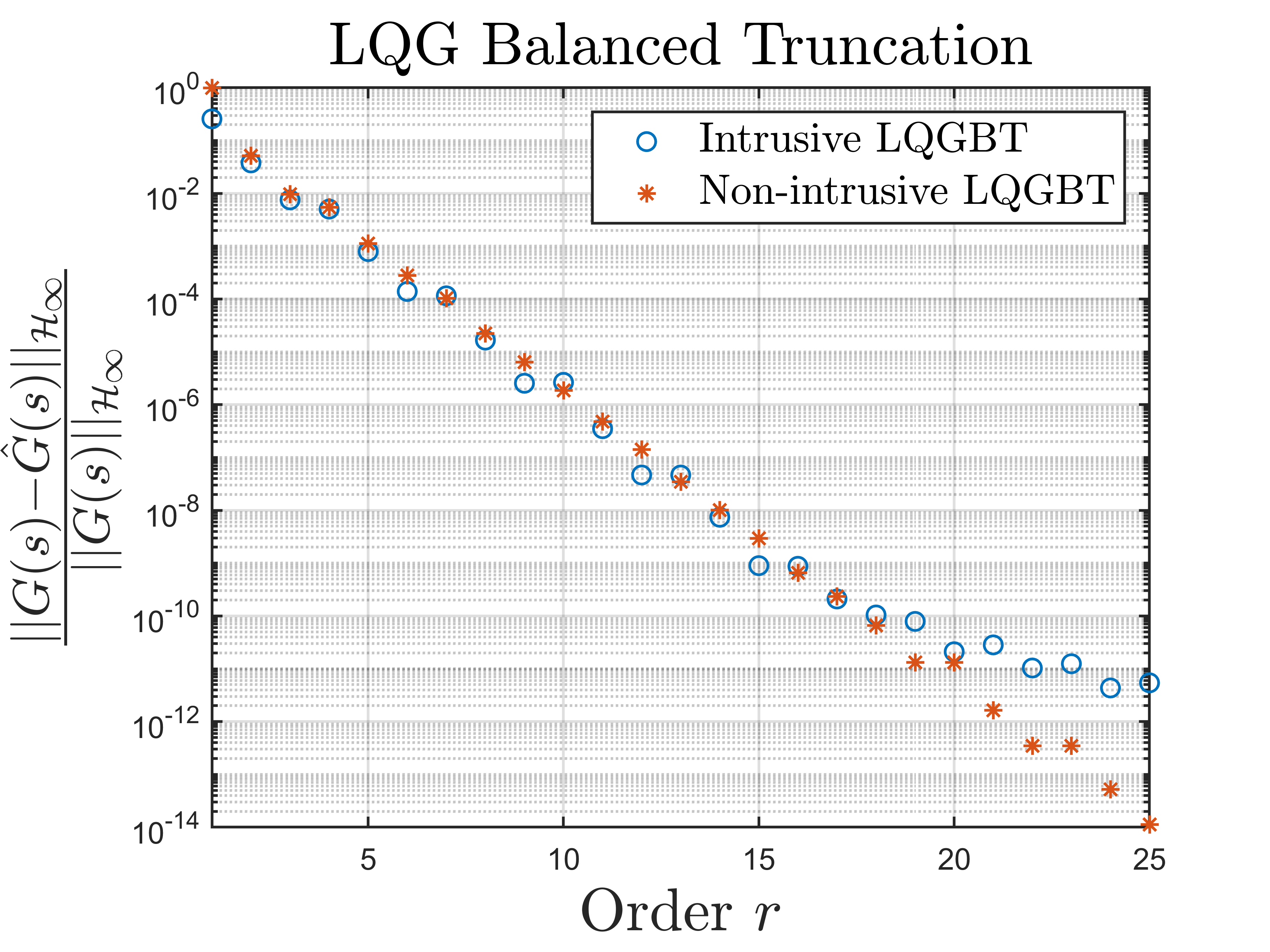}
        \caption{Relative Error Comparison}
    \end{subfigure}
    \caption{Performance Comparison between Intrusive and Non-intrusive LQGBT}\label{fig4}
\end{figure}
\begin{figure}[!h]
    \centering
    \begin{subfigure}{0.48\textwidth}
        \centering
        \includegraphics[width=\linewidth]{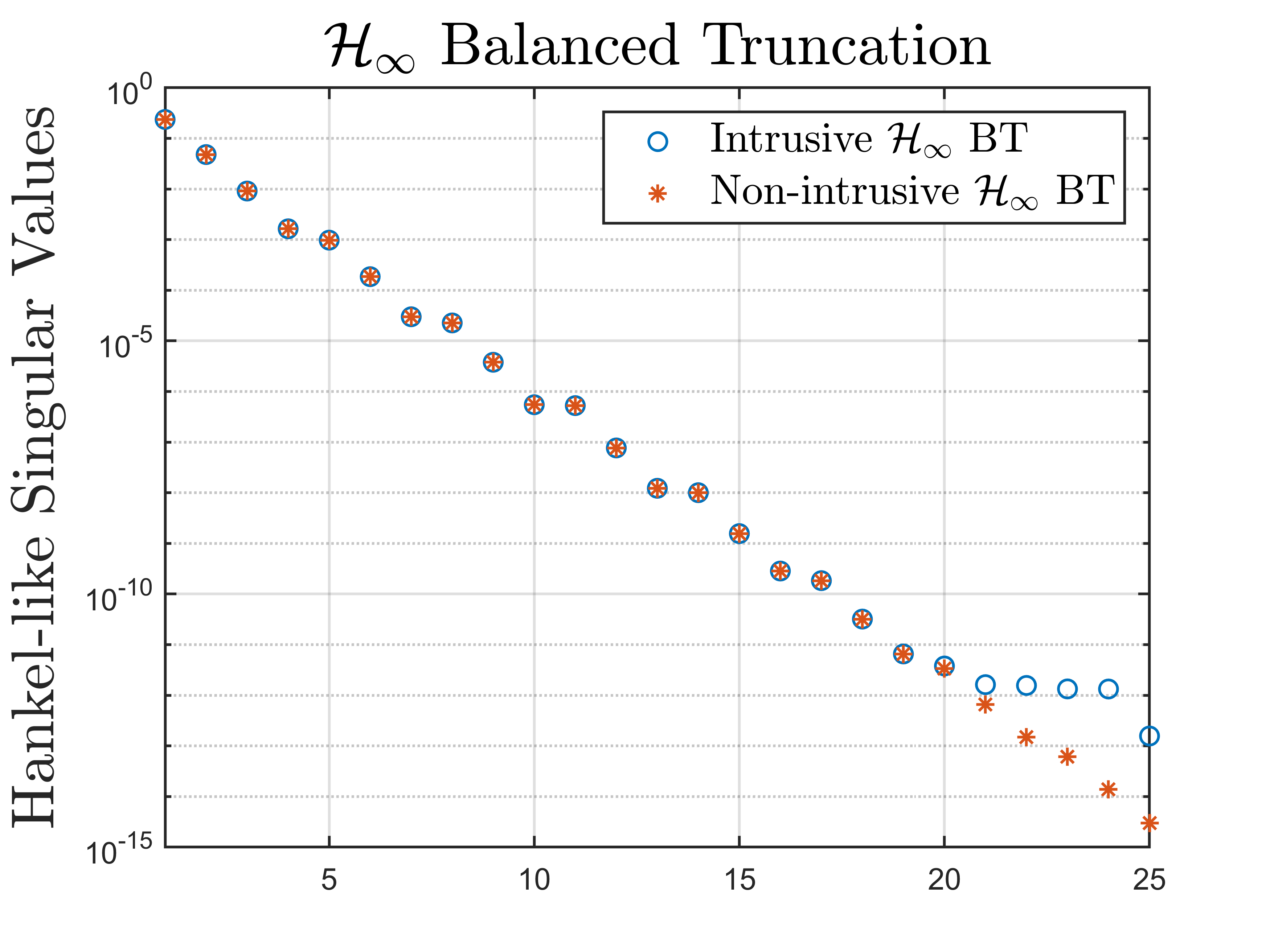}
        \caption{Hankel-like Singular Values Comparison}
    \end{subfigure}
    \hfill
    \begin{subfigure}{0.48\textwidth}
        \centering
        \includegraphics[width=\linewidth]{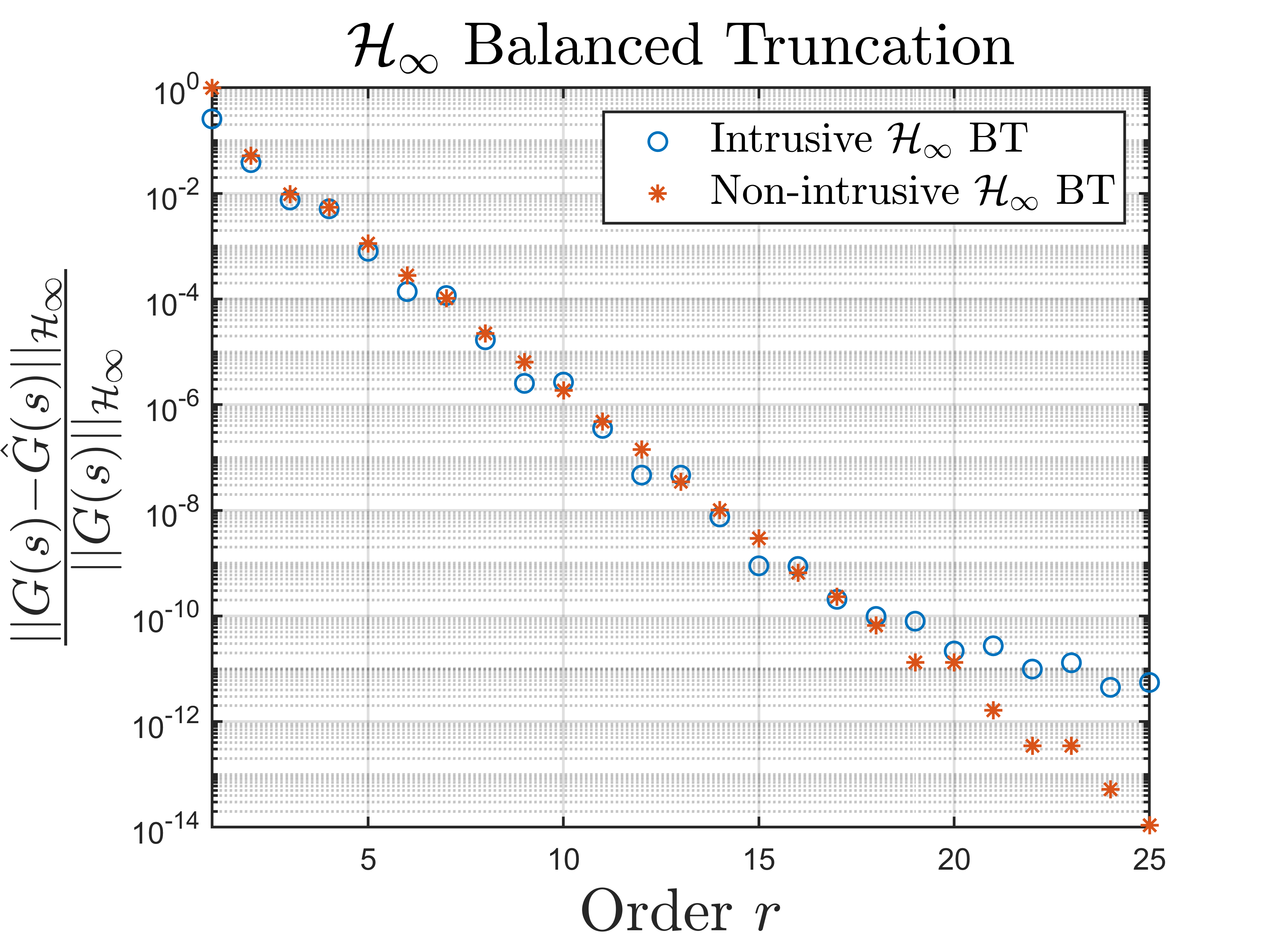}
        \caption{Relative Error Comparison}
    \end{subfigure}
    \caption{Performance Comparison between Intrusive and Non-intrusive $\mathcal{H}_\infty$ BT}\label{fig5}
\end{figure}
\begin{figure}[!h]
    \centering
    \begin{subfigure}{0.48\textwidth}
        \centering
        \includegraphics[width=\linewidth]{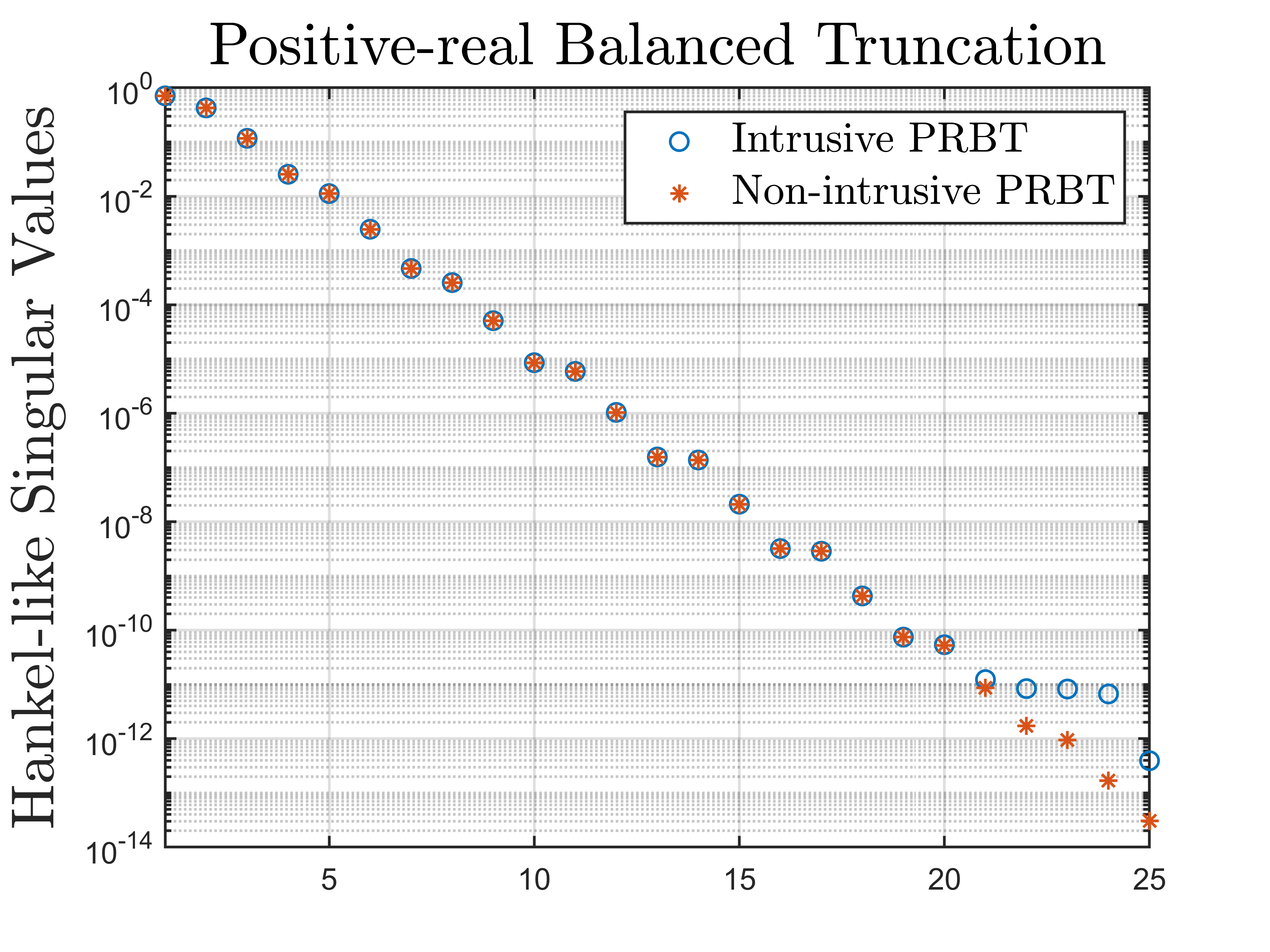}
        \caption{Hankel-like Singular Values Comparison}
    \end{subfigure}
    \hfill
    \begin{subfigure}{0.48\textwidth}
        \centering
        \includegraphics[width=\linewidth]{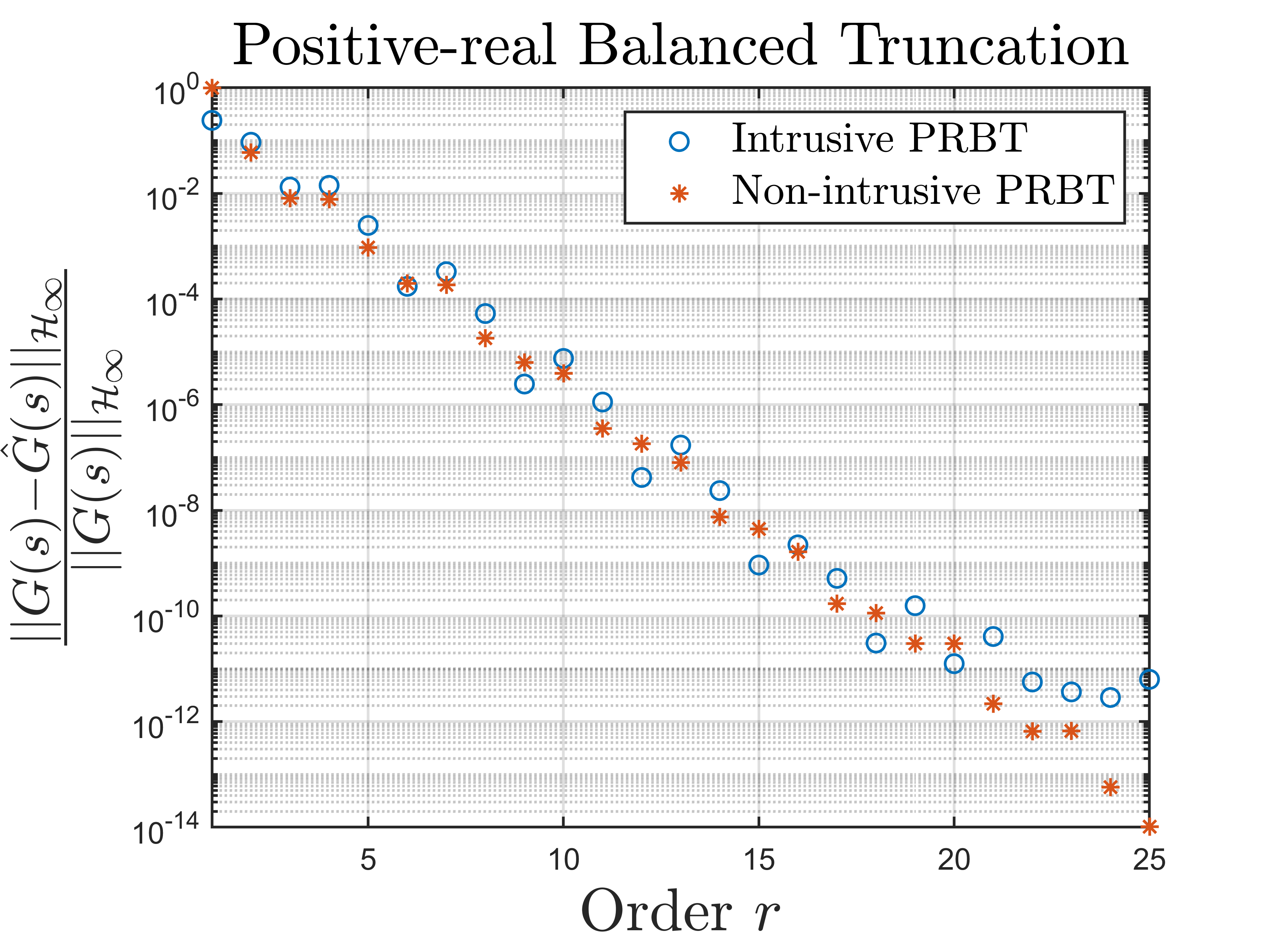}
        \caption{Relative Error Comparison}
    \end{subfigure}
    \caption{Performance Comparison between Intrusive and Non-intrusive PRBT}\label{fig6}
\end{figure}
\begin{figure}[!h]
    \centering
    \begin{subfigure}{0.48\textwidth}
        \centering
        \includegraphics[width=\linewidth]{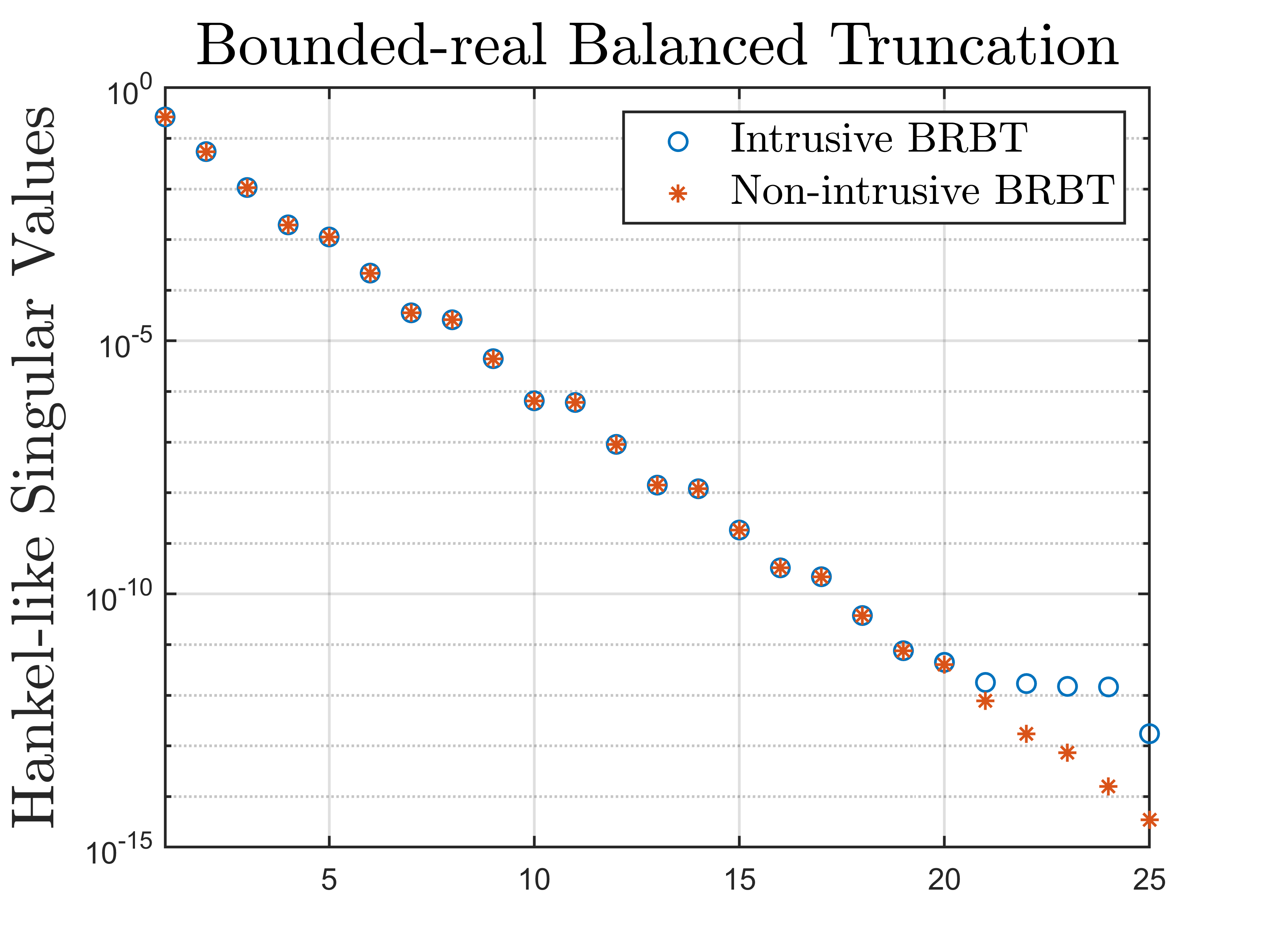}
        \caption{Hankel-like Singular Values Comparison}
    \end{subfigure}
    \hfill
    \begin{subfigure}{0.48\textwidth}
        \centering
        \includegraphics[width=\linewidth]{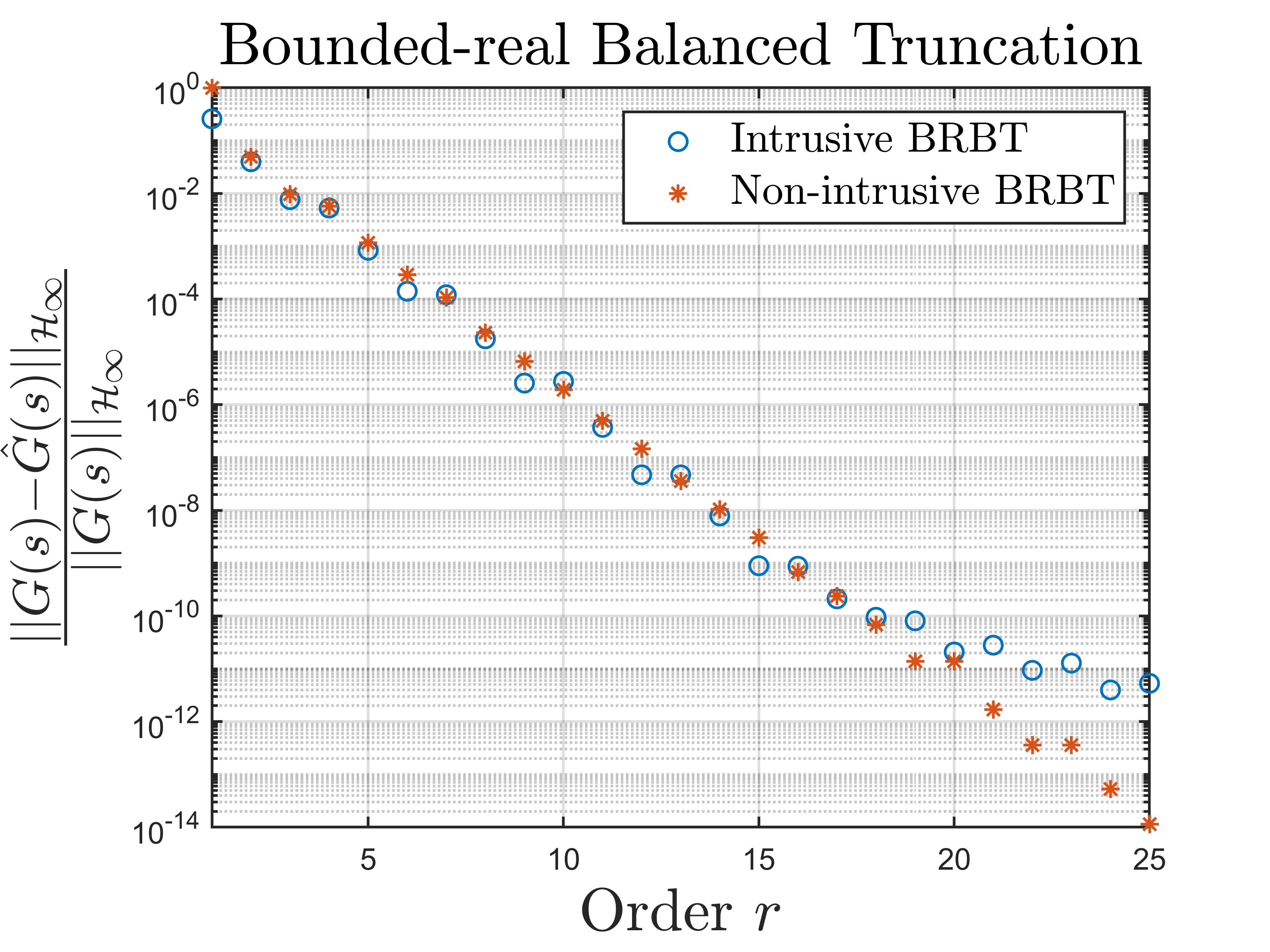}
        \caption{Relative Error Comparison}
    \end{subfigure}
    \caption{Performance Comparison between Intrusive and Non-intrusive BRBT}\label{fig7}
\end{figure}
\begin{figure}[!h]
    \centering
    \begin{subfigure}{0.48\textwidth}
        \centering
        \includegraphics[width=\linewidth]{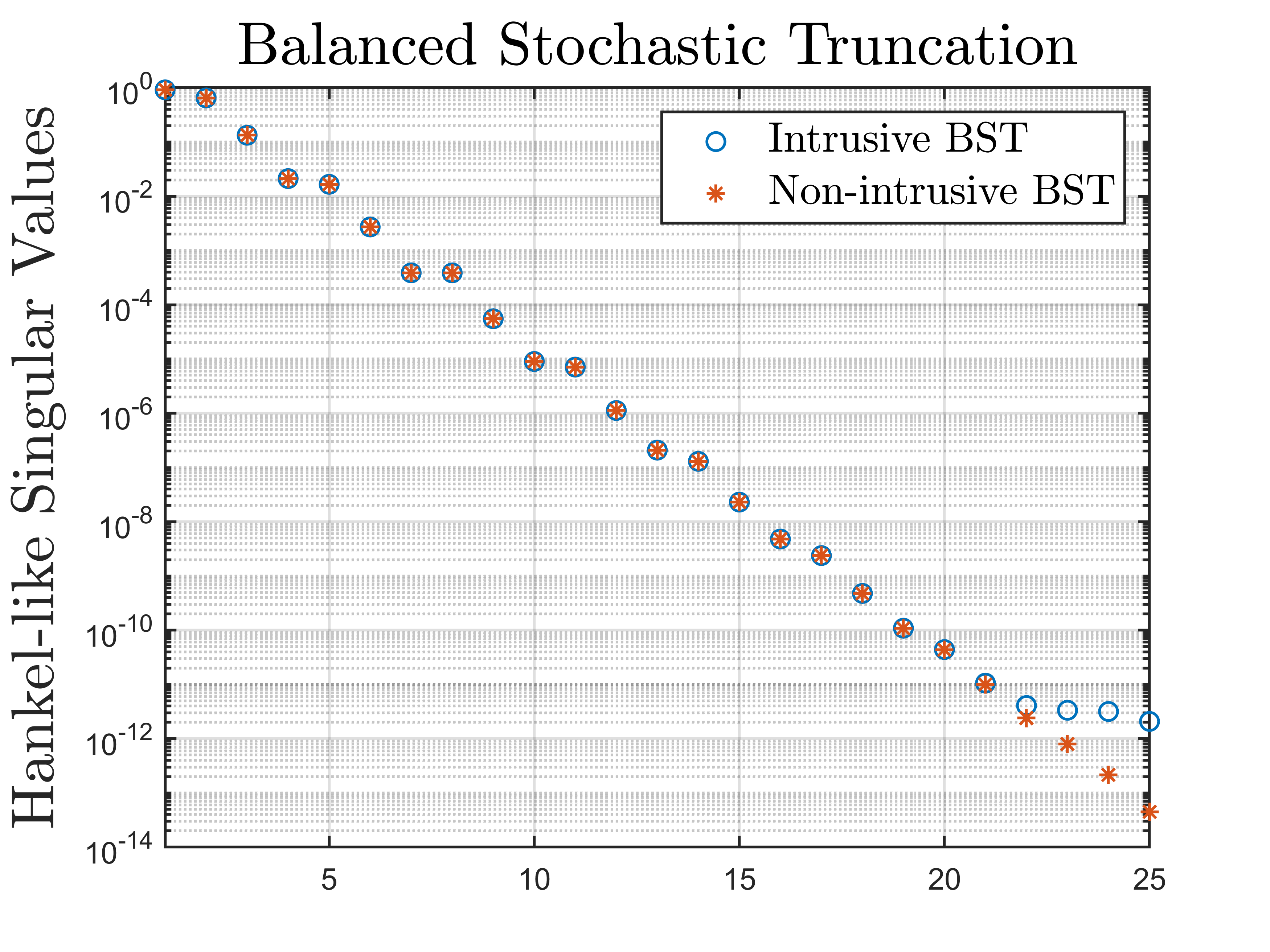}
        \caption{Hankel-like Singular Values Comparison}
    \end{subfigure}
    \hfill
    \begin{subfigure}{0.48\textwidth}
        \centering
        \includegraphics[width=\linewidth]{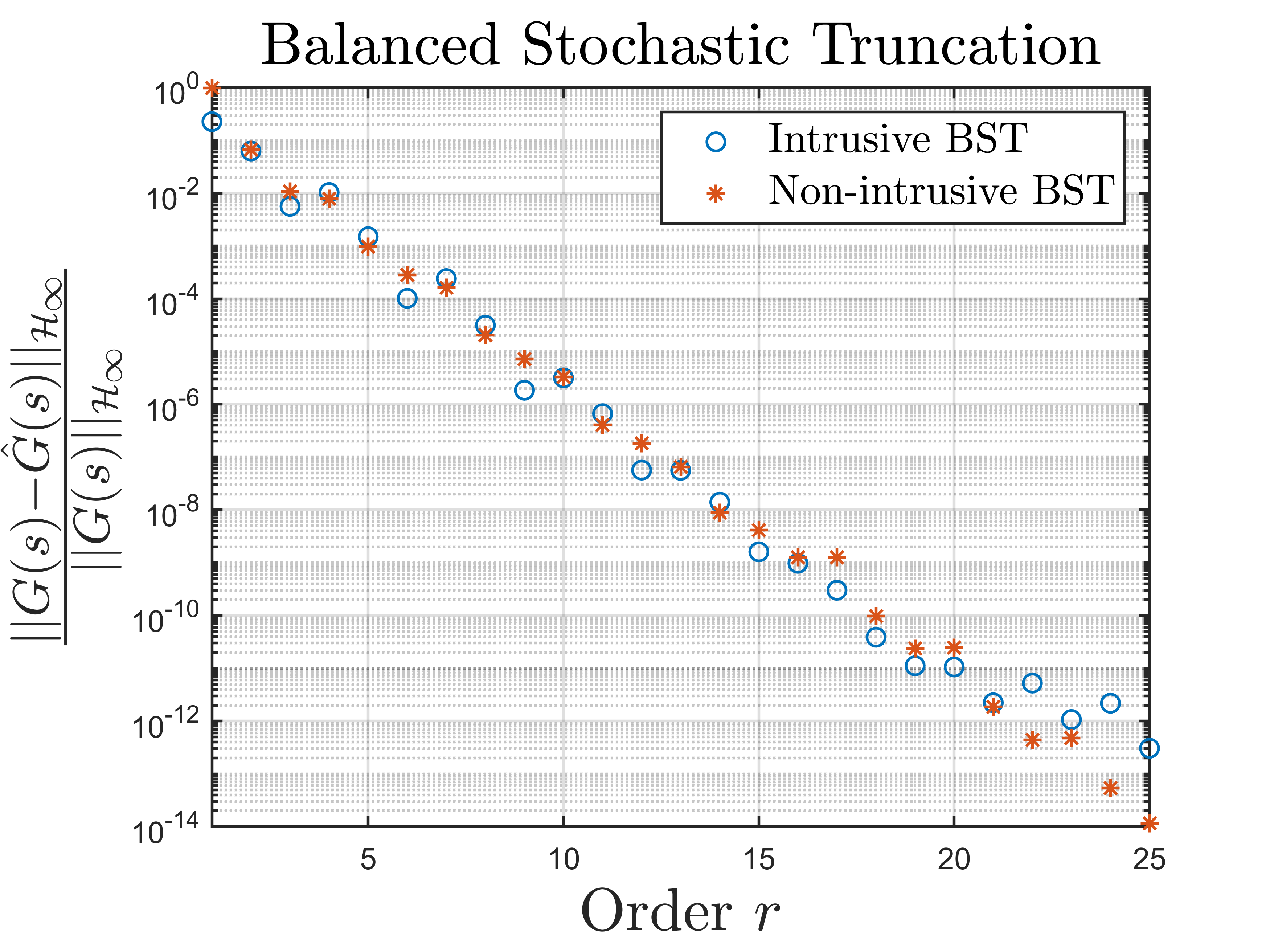}
        \caption{Relative Error Comparison}
    \end{subfigure}
    \caption{Performance Comparison between Intrusive and Non-intrusive BST}\label{fig8}
\end{figure}

The original model has a dip in the frequency domain plot at around $9.63$ rad/sec. To capture this dip, the desired frequency interval in FLBT is set to $[-30,-1]\cup[1,30]$ rad/sec. For FLBT, the sampling points $j\omega_i$ are $25$ logarithmically spaced points between $10^0$ and $10^{1.5}$, with the negative counterparts $-j\omega_i$ also included. The sampling points $j\mu_i$ are identical to $j\sigma_i$. Figure \ref{fig9} compares the Hankel-like singular values and the relative error $\frac{\|G(s)-\hat{G}(s)\|_{\mathcal{H}_\infty}}{\|G(s)\|_{\mathcal{H}_\infty}}$ for FLBT. The $25^{th}$-order ROMs from intrusive FLBT and its non-intrusive counterpart accurately capture the $20$ most dominant Hankel-like singular values. Moreover, non-intrusive FLBT achieves accuracy comparable to intrusive FLBT for ROMs of orders $1$ through $25$.
\begin{figure}[!h]
    \centering
    \begin{subfigure}{0.48\textwidth}
        \centering
        \includegraphics[width=\linewidth]{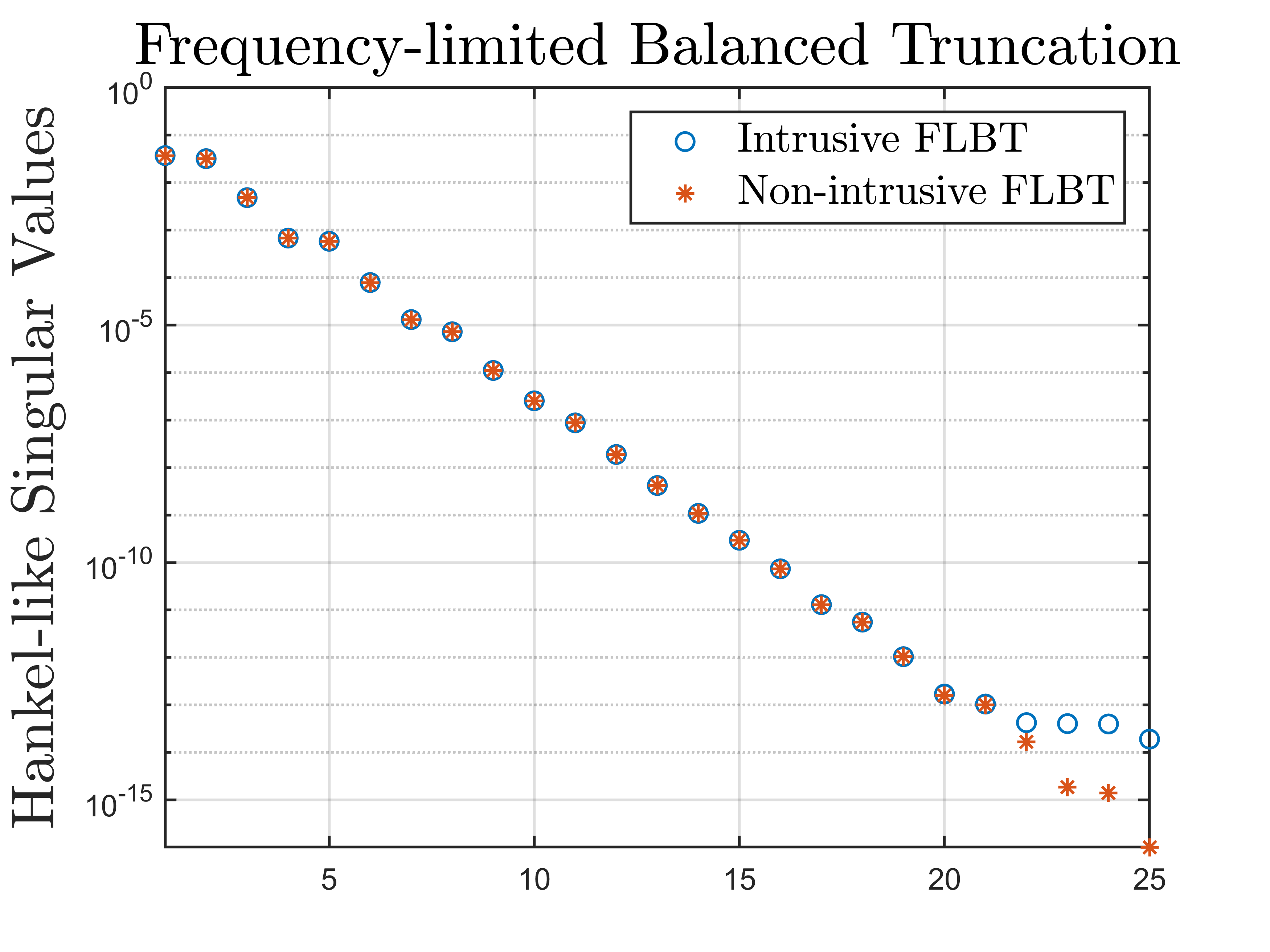}
        \caption{Hankel-like Singular Values Comparison}
    \end{subfigure}
    \hfill
    \begin{subfigure}{0.48\textwidth}
        \centering
        \includegraphics[width=\linewidth]{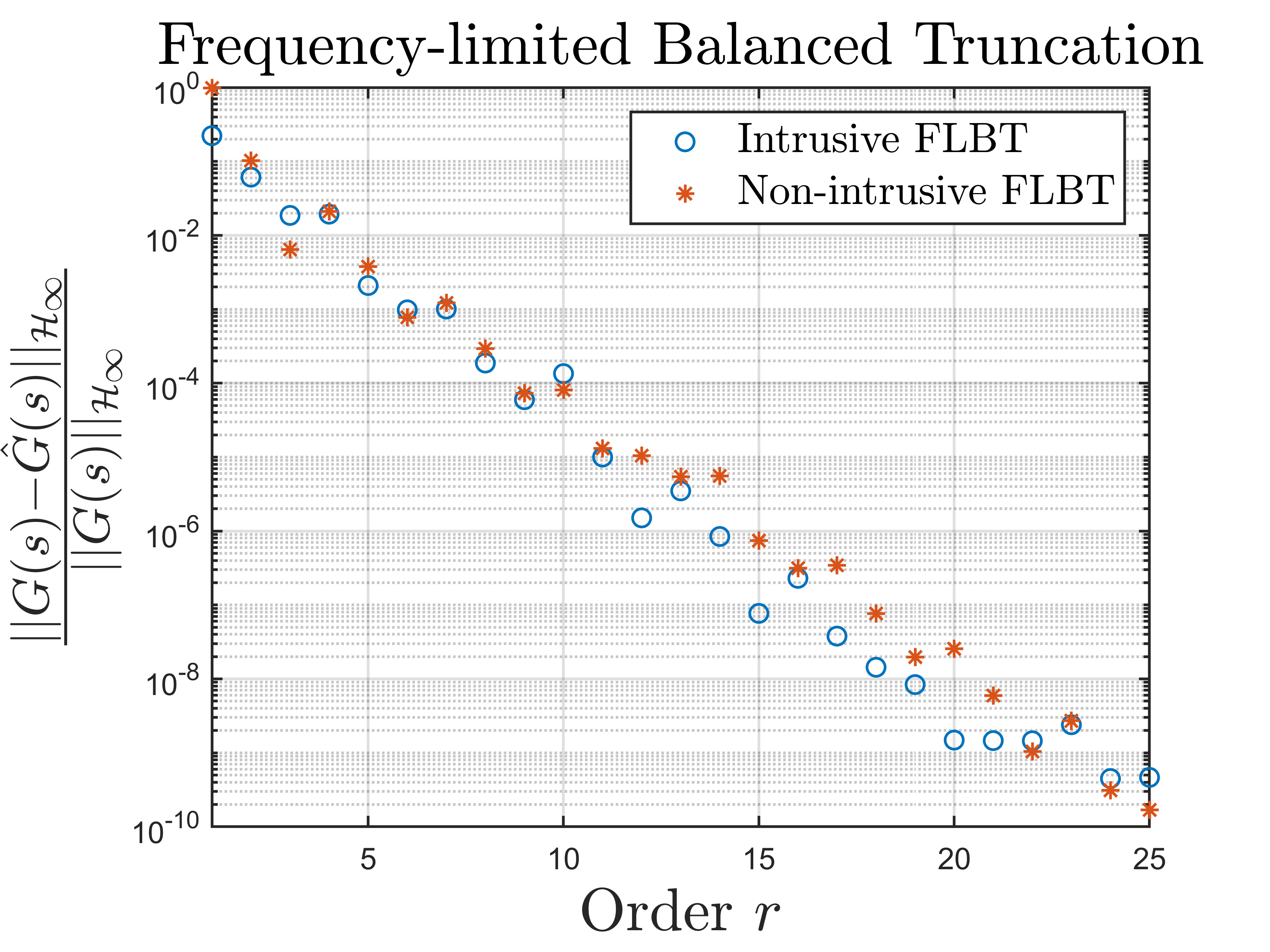}
        \caption{Relative Error Comparison}
    \end{subfigure}
    \caption{Performance Comparison between Intrusive and Non-intrusive FLBT}\label{fig9}
\end{figure}
Furthermore, the frequency domain plots of the original model and $6^{th}$-order ROMs generated by FLBT and non-intrusive FLBT are compared in Figure \ref{fig10}. The ROMs accurately capture the targeted dip in the frequency domain plot.
\begin{figure}[!h]
\centering
\includegraphics[width=\linewidth]{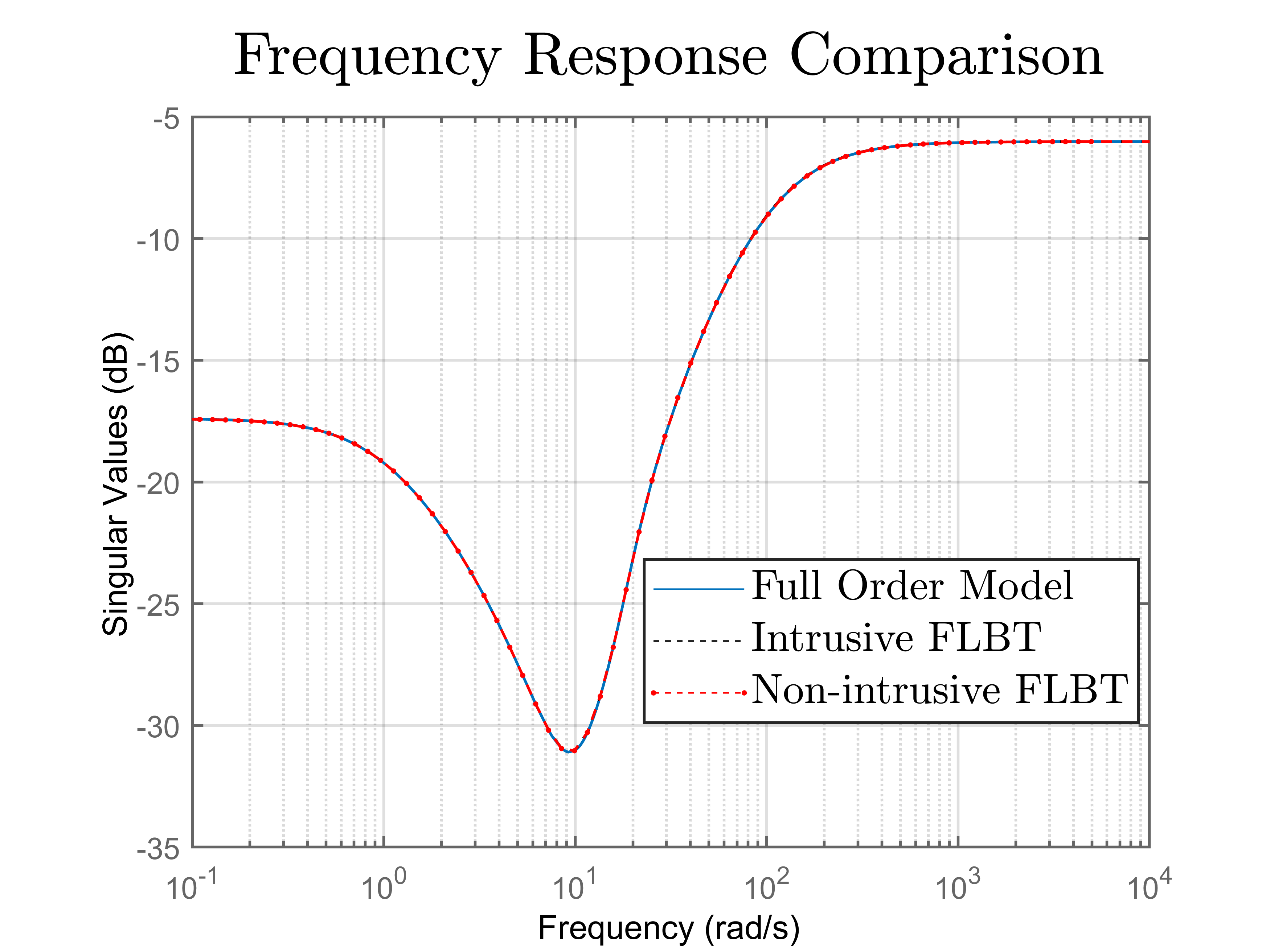}
    \caption{Frequency Domain Plot Comparison}\label{fig10}
\end{figure}
\section{Conclusion}
This paper introduces a projection-based framework for the non-intrusive, data-driven implementation of BT and eight of its generalizations. By constructing projections directly from transfer function samples on the imaginary axis, the proposed approach circumvents the need for spectral factor measurements or right-half-plane samples required by previous methods \cite{zulfiqar2025data,reiter2023generalizations}. The key insight is to approximate Gramians implicitly through projection rather than via numerical quadrature. The proposed framework avoids explicitly solving any projected matrix equations. Instead, the projected matrix equations are approximated via analytical expressions due to block-diagonally-dominant solutions. Numerical results demonstrate that the proposed data-driven implementations achieve accuracy comparable to intrusive methods and accurately capture dominant Hankel-like singular values across all considered balancing techniques.
\section*{Appendix A: Proof for Diagonal Dominance of $X_p$}
\begin{proof}
The Sylvester equation \eqref{Xp} has a unique solution because the spectra of $S_p^{*}$ and $-S_v$ are disjoint: $\operatorname{spec}(S_p^{*}) = \{\epsilon - j\omega_i\}_{i=1}^{n_s}$ (each eigenvalue with multiplicity $m$) and $\operatorname{spec}(-S_v) = \{-j\omega_i\}_{i=1}^{n_s}$ (each with multiplicity $m$) satisfy $\operatorname{Re}(\epsilon - j\omega_i) = \epsilon > 0 = \operatorname{Re}(-j\omega_i)$ for all $i$.

Since $S_v$, $S_p^{*}$, and $L_v^TL_v = \mathbf{1}_{n_s}\mathbf{1}_{n_s}^{\mathsf{T}} \otimes I_m$ share the Kronecker structure $\cdot \otimes I_m$, we seek a solution of the form $X_p = Y \otimes I_m$ with $Y \in \mathbb{C}^{n_s \times n_s}$. Substituting into \eqref{Xp} and exploiting the mixed-product property of Kronecker products yields the reduced equation
\begin{equation}
\operatorname{diag}(\epsilon - j\omega_1, \dots, \epsilon - j\omega_{n_s}) \, Y + Y \, \operatorname{diag}(j\omega_1, \dots, j\omega_{n_s}) = \mathbf{1}_{n_s}\mathbf{1}_{n_s}^{\mathsf{T}}.
\end{equation}
Entry-wise for $i,j = 1,\dots,n_s$:
\begin{equation}
(\epsilon - j\omega_i) Y_{ij} + Y_{ij} (j\omega_j) = 1
\;\Longrightarrow\;
Y_{ij} = \frac{1}{\epsilon + j(\omega_j - \omega_i)},
\end{equation}
which is well-defined for all $\epsilon > 0$ since $\operatorname{Re}(\epsilon + j(\omega_j - \omega_i)) = \epsilon > 0$.

For diagonal entries ($i=j$), $Y_{ii} = 1/\epsilon$. For off-diagonal entries ($i \neq j$),
\begin{equation}
|Y_{ij}| = \frac{1}{\sqrt{\epsilon^2 + (\omega_j - \omega_i)^2}} \leq \frac{1}{|\omega_j - \omega_i|} \leq \frac{1}{\Delta_{\min}}.
\end{equation}

Define $E := Y - \frac{1}{\epsilon}I_{n_s}$, so $E_{ii}=0$ and $E_{ij}=Y_{ij}$ for $i \neq j$. Using the Frobenius norm identity $\|A \otimes B\|_F = \|A\|_F \|B\|_F$,
\begin{align}
\|X_p - \widetilde{X}_p\|_F &= \|E \otimes I_m\|_F = \|E\|_F \, \|I_m\|_F = \|E\|_F \sqrt{m}, \\
\|E\|_F^2 &= \sum_{i \neq j} |Y_{ij}|^2 \leq \sum_{i \neq j} \frac{1}{\Delta_{\min}^2} = \frac{n_s(n_s-1)}{\Delta_{\min}^2}
\;\Longrightarrow\;\nonumber\\
&\hspace*{2cm}\|E\|_F \leq \frac{\sqrt{n_s(n_s-1)}}{\Delta_{\min}}, \\
\|\widetilde{X}_p\|_F &= \frac{1}{\epsilon} \|I_{mn_s}\|_F = \frac{\sqrt{mn_s}}{\epsilon}.
\end{align}
Therefore,
\begin{equation}
\frac{\|X_p - \widetilde{X}_p\|_F}{\|\widetilde{X}_p\|_F}
\leq \frac{\sqrt{n_s(n_s-1)}/\Delta_{\min} \cdot \sqrt{m}}{\sqrt{mn_s}/\epsilon}
= \frac{\epsilon \sqrt{n_s-1}}{\Delta_{\min}},
\end{equation}
proving \eqref{eq:relative_error}. The tolerance condition follows immediately.

For diagonal dominance, observe that for each row $i$,
\begin{equation}
|Y_{ii}| - \sum_{j \neq i} |Y_{ij}| \geq \frac{1}{\epsilon} - \sum_{j \neq i} \frac{1}{|\omega_j - \omega_i|}
\geq \frac{1}{\epsilon} - \frac{n_s-1}{\Delta_{\min}}.
\end{equation}
This quantity is positive when $\epsilon < \Delta_{\min}/(n_s-1)$, establishing strict diagonal dominance of $Y$ and hence of $X_p = Y \otimes I_m$.
\end{proof}
\section*{Appendix B: Proof for Block Diagonal Dominance of $X_z$}
\begin{proof}
Because $S_z^{*}$ and $S_v$ are block-diagonal with scalar multiples of $I_m$, the Sylvester equation decouples into independent $m\times m$ block equations. Direct computation using $L_v = \mathbf{1}_{n_s}^{\mathsf{T}} \otimes I_m$ shows that the $(i,j)$-block of the right-hand side equals $M_j$ (independent of $i$). Hence each block satisfies
\begin{equation*}
\bigl(\epsilon + j(\omega_j - \omega_i)\bigr) X_{ij} = M_j,
\end{equation*}
and since the scalar coefficient is nonzero for all $i,j$ (as $\epsilon>0$ and frequencies are distinct), the unique solution is
\begin{equation*}
X_{ij} = \frac{M_j}{\epsilon + j(\omega_j - \omega_i)}.
\end{equation*}
For diagonal blocks ($i=j$), $X_{ii} = M_i/\epsilon$; for off-diagonal blocks ($i\neq j$),
\begin{equation*}
\|X_{ij}\| = \frac{\|M_j\|}{\sqrt{\epsilon^{2} + (\omega_j - \omega_i)^{2}}}.
\end{equation*}
The block diagonal dominance condition for row $i$ is therefore
\begin{equation*}
\frac{\|M_i\|}{\epsilon} > \sum_{j\neq i} \frac{\|M_j\|}{\sqrt{\epsilon^{2} + (\omega_j - \omega_i)^{2}}}.
\end{equation*}
Multiplying by $\epsilon>0$ and using $\sqrt{\epsilon^{2} + (\omega_j - \omega_i)^{2}} \geq |\omega_j - \omega_i| \geq \Delta_{\mathrm{min}}$ for $i\neq j$ yields the sufficient condition
\begin{equation*}
\|M_i\| > \frac{\epsilon}{\Delta_{\mathrm{min}}} \sum_{j\neq i} \|M_j\|,
\end{equation*}
which is equivalent to
\begin{equation*}
\epsilon < \Delta_{\mathrm{min}} \cdot \frac{\|M_i\|}{\displaystyle\sum_{j\neq i} \|M_j\|}.
\end{equation*}
Requiring this for all rows $i$ gives the stated bound on $\epsilon$. The bound on the relative off-diagonal weight follows immediately from the same inequalities.
\end{proof}
\section*{Appendix C: Proof for Block Diagonal Dominance of $X_{sp}$}
\begin{proof}
The Sylvester equation admits the block-wise representation: for each $1 \leq i,j \leq n_s$,
\begin{equation}
\bigl(\epsilon + j(\omega_j - \omega_i)\bigr) X_{ij} + \sum_{k=1}^{n_s} X_{ik} Q_{kj} = I_m.
\label{eq:block_sylvester}
\end{equation}
Define $\Delta_{ij} := \epsilon + j(\omega_j - \omega_i)$. By construction,
\[
|\Delta_{ii}| = \epsilon, \qquad |\Delta_{ij}| = \sqrt{\epsilon^2 + (\omega_j - \omega_i)^2} \geq |\omega_j - \omega_i| \geq \Delta_{\min} \quad \text{for } i \neq j.
\label{eq:delta_bounds}
\]

Taking spectral norms in \eqref{eq:block_sylvester} and applying the triangle inequality yields
\[
|\Delta_{ij}| \, \|X_{ij}\|_2 \leq 1 + \sum_{k=1}^{n_s} \|X_{ik}\|_2 \, \|Q_{kj}\|_2.
\]
Using the assumptions $\|Q_{kj}\|_2 \leq \eta \epsilon$ for all $k,j$, so
\begin{equation}
|\Delta_{ij}| \, \|X_{ij}\|_2 \leq 1 + \eta \epsilon \sum_{k=1}^{n_s} \|X_{ik}\|_2.
\label{eq:norm_ineq}
\end{equation}

For diagonal blocks ($i=j$), \eqref{eq:norm_ineq} and \eqref{eq:delta_bounds} give
\[
\epsilon \|X_{ii}\|_2 \leq 1 + \eta \epsilon r_i, \qquad r_i := \sum_{k=1}^{n_s} \|X_{ik}\|_2.
\]
Hence
\begin{equation}
\|X_{ii}\|_2 \leq \frac{1}{\epsilon} + \eta r_i.
\label{eq:diag_upper}
\end{equation}
For off-diagonal blocks ($i \neq j$),
\begin{equation}
\|X_{ij}\|_2 \leq \frac{1}{\Delta_{\min}} + \frac{\eta \epsilon}{\Delta_{\min}} r_i.
\label{eq:offdiag_upper}
\end{equation}
Summing \eqref{eq:offdiag_upper} over $j \neq i$,
\begin{equation}
\sum_{\substack{j=1 \\ j \neq i}}^{n_s} \|X_{ij}\|_2 \leq \frac{n_s-1}{\Delta_{\min}} + \frac{(n_s-1)\eta\epsilon}{\Delta_{\min}} r_i.
\label{eq:offdiag_sum}
\end{equation}
Combining $r_i = \|X_{ii}\|_2 + \sum_{j \neq i} \|X_{ij}\|_2$ with \eqref{eq:diag_upper} and \eqref{eq:offdiag_sum},
\[
r_i \leq \frac{1}{\epsilon} + \eta r_i + \frac{n_s-1}{\Delta_{\min}} + \frac{(n_s-1)\eta\epsilon}{\Delta_{\min}} r_i.
\]
Rearranging and using assumption (3) ($\epsilon \leq \Delta_{\min}/(2n_s)$ implies $(n_s-1)\epsilon/\Delta_{\min} \leq 1/2$),
\[
r_i \Bigl(1 - \eta - \tfrac{1}{2}\eta\Bigr) \leq \frac{1}{\epsilon} + \frac{n_s-1}{\Delta_{\min}}
\quad \Longrightarrow \quad
r_i \bigl(1 - \tfrac{3}{2}\eta\bigr) \leq \frac{1}{\epsilon} + \frac{n_s-1}{\Delta_{\min}}.
\]
Since $\eta \leq 1/(6n_s) \leq 1/12$, we have $1 - \tfrac{3}{2}\eta \geq 7/8$. Moreover, $\epsilon \leq \Delta_{\min}/(2n_s)$ implies $(n_s-1)/\Delta_{\min} \leq 1/(2\epsilon)$. Thus,
\begin{equation}
r_i \leq \frac{8}{7} \Bigl(\frac{1}{\epsilon} + \frac{1}{2\epsilon}\Bigr) = \frac{12}{7\epsilon} < \frac{2}{\epsilon}.
\label{eq:ri_bound}
\end{equation}
Consequently,
\begin{equation}
\eta \epsilon r_i < 2\eta \leq \frac{1}{3n_s} \leq \frac{1}{6} < 1.
\label{eq:small_pert}
\end{equation}

Rewriting \eqref{eq:block_sylvester} for $i=j$ as $\epsilon X_{ii} = I_m - \sum_{k=1}^{n_s} X_{ik} Q_{ki}$ and applying the reverse triangle inequality with \eqref{eq:small_pert},
\begin{equation}
\|X_{ii}\|_2 \geq \frac{1}{\epsilon} \Bigl(1 - \sum_{k=1}^{n_s} \|X_{ik}\|_2 \, \|Q_{ki}\|_2 \Bigr) \geq \frac{1}{\epsilon} (1 - \eta \epsilon r_i) > \frac{1 - 2\eta}{\epsilon}.
\label{eq:diag_lower}
\end{equation}
For the off-diagonal sum, combining \eqref{eq:offdiag_sum} and \eqref{eq:ri_bound},
\begin{equation}
\sum_{\substack{j=1 \\ j \neq i}}^{n_s} \|X_{ij}\|_2 < \frac{n_s}{\Delta_{\min}} (1 + 2\eta).
\label{eq:offdiag_final}
\end{equation}
Using \eqref{eq:diag_lower}, \eqref{eq:offdiag_final}, and $1/\epsilon \geq 2n_s/\Delta_{\min}$ (from assumption (3)),
\[
\|X_{ii}\|_2 - \sum_{\substack{j \neq i}} \|X_{ij}\|_2 > \frac{1 - 2\eta}{\epsilon} - \frac{n_s(1 + 2\eta)}{\Delta_{\min}}
\geq \frac{2n_s(1 - 2\eta)}{\Delta_{\min}} - \frac{n_s(1 + 2\eta)}{\Delta_{\min}}
= \frac{n_s}{\Delta_{\min}} (1 - 6\eta).
\]
Since $\eta \leq 1/(6n_s)$ and $n_s \geq 2$, we have $1 - 6\eta \geq 1 - 1/n_s \geq 1/2 > 0$. As $\Delta_{\min} > 0$, the right-hand side is strictly positive, establishing strict block diagonal dominance for all $i$.
\end{proof}
%\bibliography{mybibfile}

\end{document}